\documentclass[10pt]{article}
\usepackage{fullpage}
\usepackage[hidelinks]{hyperref}
\usepackage{booktabs}
\usepackage{amsfonts}

\usepackage{authblk}
\usepackage{amssymb}
\usepackage{tikz}
\usepackage{amstext}
\usepackage{amsmath}
\usepackage{xspace}
\usepackage{theorem}
\usepackage{graphicx}
\usepackage{mdframed}
\usepackage{comment}
\usepackage{mathdots}
\usepackage{xcolor}
\usepackage{graphics}
\usepackage{colordvi}
\usepackage{color}
\usepackage{hyperref}
\usepackage{paralist}
\usepackage{bm}
\usepackage{mathpazo}
\usepackage[boxruled,lined,linesnumbered]{algorithm2e}
\usepackage{cleveref}

\usepackage{float}
\usepackage{kbordermatrix}

\usepackage{lineno}

\usepackage{subcaption}
\newtheorem{theorem}{Theorem}

\newtheorem{lemma}{Lemma}
\newtheorem{observation}{Observation}

\newtheorem{corollary}{Corollary}

\newtheorem{proposition}[theorem]{Proposition}

\newtheorem{definition}{Definition}

\newtheorem{remark}{Remark}

\newcommand{\pset}{{\mathcal P}}

\newcommand{\opt}{{\sf OPT}\xspace}

\newcommand{\dg}{{\sf CG}\xspace}

\newcommand{\FF}{\textsc{F}}

\newcommand{\GG}{\mathcal{G}}
\newcommand{\GR}{\textsc{Greedy}\xspace}
\newcommand{\GGR}{\textsc{Ggreedy}\xspace}
\newcommand{\WBO}{\textsc{Wilber1}}
\newcommand{\WBT}{\textsc{Wilber2}}

\newcommand{\LE}{\text{left}}
\newcommand{\RI}{\text{right}}

\newcommand{\UB}{\textsc{Ub}}

\newcommand{\PR}{\textsc{Preferred-Child}}

\newcommand{\OPT}{\textsc{Opt}}

\newcommand{\TT}{\mathcal{T}}
\newcommand{\HTT}{\hat{\mathcal{T}}}

\newcommand{\SQLN}{{\sqrt{\log n}}}
\newcommand{\EL}{\ell}

\newcommand{\WS}{\textsc{Ws}}

\newcommand{\CC}{\mathcal{C}}
\newcommand\numberthis{\addtocounter{equation}{1}\tag{\theequation}}
\newcommand{\aset}{{\mathcal A}}

\newcommand{\nset}{{\mathcal N}}
\newcommand{\wset}{{\mathcal W}}

\newcommand{\tset}{{\mathcal T}}
\newcommand{\gab}{{\sf GAB}}

\newcommand{\attention}[1]{\textcolor{red}{*** #1 *** }}

\newenvironment{proof}{\par \noindent{\bf Proof:}}{\hfill\stopproof}
\def\stopproof{\square}
\def\square{\vbox{\hrule height.2pt\hbox{\vrule width.2pt height5pt \kern5pt
\vrule width.2pt} \hrule height.2pt}}

        {\hspace*{\fill}$\Box$\par\vspace{4mm}}

\usepackage{authblk}
\title{The Group Access Bounds for Binary Search Trees}

\date{}

\begin{document}

\author[$\ast$]{Parinya Chalermsook\thanks{Supported by European Research Council (ERC) under the
European Union’s Horizon 2020 research and innovation programme (grant
agreement No. 759557)}}
\author[$\dagger$]{Manoj Gupta}
\author[$\ast$]{Wanchote Jiamjitrak}
\author[$\dagger$]{Akash Pareek\thanks{Part of the research was done when the author was visiting Aalto University, Finland.}}
\author[$\ddagger$]{Sorrachai Yingchareonthawornchai}
\affil[$\ast$]{Aalto University}
\affil[$\dagger$]{IIT Gandhinagar}
\affil[$\ddagger$]{Simons Institute for the Theory of Computing, UC Berkeley}

\renewcommand\Authands{ and } 

\maketitle

\begin{abstract}

The access lemma (Sleator and Tarjan, JACM  1985) is a property of binary search trees (BSTs) that implies interesting consequences such as static optimality, static finger, and working set property. However, there are known corollaries of the dynamic optimality that cannot be derived via the access lemma, such as the dynamic finger, and any $o(\log n)$-competitive ratio to the optimal BST where $n$ is the number of keys. 

In this paper, we introduce the {\em group access bound} that can be defined with respect to a reference \textit{group access tree}. 
Group access bounds generalize the access lemma and imply properties that are far stronger than those implied by the classical access lemma. For each of the following results, there is a group access tree whose group access bound
\begin{enumerate}

    \item Is $O(\sqrt{\log n})$-competitive to the optimal BST. 


    \item Achieves the $k$-finger bound with an \textit{additive} term of $O(m \log k \log \log  n)$ (randomized) when the reference tree is an almost complete binary tree.
    
	\item Satisfies the unified bound with an \textit{additive} term of $O(m \log \log n)$.
 
	\item Matches the unified bound with a time window $k$ with an \textit{additive} term of $O(m \log k \log \log n)$ (randomized).
     
\end{enumerate}

Furthermore, we prove the simulation theorem: For every group access tree, there is an online BST algorithm that is $O(1)$-competitive with its group access bound. In particular, any new group access bound will automatically imply a new BST algorithm achieving the same bound.  Thereby, we obtain an improved $k$-finger bound (reference tree is an almost complete binary tree), an improved unified bound with a time window $k$, and matching the best-known bound for Unified bound in the BST model. 
Since any dynamically optimal BST must achieve the group access bounds, we believe our results provide a new direction towards proving $o(\log n)$-competitiveness of Splay tree and Greedy, two prime candidates for the dynamic optimality conjecture.

\end{abstract}

\newpage

\section{Introduction}
\label{sec:intro}

In the amortized analysis of a self-adjusting binary search tree (BST), the \textit{access lemma}~\cite{SleatorT85} is perhaps the most fundamental property of a BST algorithm that allows us to prove the competitiveness against many performance benchmarks, including temporal and spatial locality. If a BST algorithm satisfies the access lemma, then, by plugging in appropriate parameters,  the algorithm also satisfies many interesting corollaries of the \textit{dynamic optimality conjecture} including, for example, {\em balance theorem~\cite{SleatorT85}, static optimality~\cite{SleatorT85,fox11}, static finger property~\cite{SleatorT85},  working set property~\cite{SleatorT85}, and key-independent optimality~\cite{Iacono_optimality}}. For example, Splay tree~\cite{SleatorT85}, Greedy~\cite{demaine_geometry,munro2000competitiveness}, and Multi-splay tree~\cite{multi-splay-thesis2006} are all known to satisfy the access lemma~\cite{SleatorT85,fox11,multi-splay-thesis2006}.

Despite these many applications, several strong BST properties cannot be implied via the access lemma, including \textit{``non-trivial'' competitiveness}, $k$-\textit{finger property} (or even the (weaker) dynamic finger property~\cite{cole2000dynamic}), and the \textit{unified bound}~\cite{iacon_unified,derryberry2009skip}.  
For completeness, we discuss each one of them in turn.   
\paragraph{Competitiveness.} Dynamic optimality conjecture~\cite{SleatorT85} postulates that there is an online binary search tree (BST) on $n$ keys whose cost to perform a search (access) sequence $(x_1, \ldots, x_m) \in \{1,2,\ldots, n\}^m$ (including the cost to adjust the internal structure in between the sequence) is at most that of the offline optimum up to a constant factor. We say that a BST algorithm is $f(n)$-\textit{competitive} if its total cost is, at most, the cost of the offline optimum up to a factor of $f(n)$. A BST algorithm is dynamically optimal if it is $O(1)$-competitive.

Splay tree~\cite{SleatorT85} and Greedy~\cite{demaine_geometry} are widely regarded as the prime candidates for dynamic optimality conjecture. However, the best-known competitiveness of both the algorithms remains $O(\log n)$, which can be shown from the access lemma (via static optimality) or using any balanced trees. The access lemma cannot be used to prove $o(\log n)$-competitive (for example, we cannot even derive the \textit{sequential access theorem}~\cite{tarjan1985sequential} via the access lemma). In contrast, there are BST algorithms with $O(\log \log n)$-competitiveness. In \cite{tango_trees}, the authors presented the first  $O(\log \log n)$-competitive binary search tree, which they call Tango Trees. Several subsequent results  provided alternate $O(\log \log n)$-competitive algorithms~\cite{georgakopoulos2008chain,wang2006log,bose2010log,chalermsook2020new}. Despite achieving the best-known competitive ratios,  the fact that these algorithms do not satisfy many corollaries of the dynamic optimality conjecture makes them less promising than Splay and Greedy. 


\paragraph{$k$-Finger Bound.} 
Several notions of finger bounds~\cite{SleatorT85,cole2000dynamic,IaconoL16} have been introduced to capture the ``locality'' of input sequences. The strongest finger bound, called $k$-Finger bound,  was motivated by the connection between BSTs and the $k$-server problem~\cite{chalermsook2018multi}. It has still remained unclear whether any online BST algorithm achieves this property. 

We define $k$-Finger bound as follows. Assume we have an almost complete binary tree where each leaf represents a distinct key from $\{1,\dots n\}$. This tree is called the {\em reference tree}. Assume that there are $k$-fingers stationed at $k$ arbitrary leaves in the reference tree.

\begin{definition} (\textbf{$k$-Finger bound}): When a key $x_t$ is searched at time $t$, we must move one finger from its current position to the node containing $x_t$. 
The cost of the search is the number of nodes on the unique path connecting the finger's source to its destination. We define $\FF^k(X)$ as the minimum, overall finger movement strategies distance traversed by the fingers to process the sequence $X$ when the reference tree is an almost complete binary tree. 
\end{definition}



 The classical access lemma cannot imply the $k$-finger property even when $k = 1$ (called the lazy finger bound~\cite{IaconoL16}, which generalizes the dynamic finger bound~\cite{cole2000dynamic}). Nonetheless, it is possible to prove a non-trivial bound w.r.t. $k$-Finger  using different techniques. In \cite{chalermsook2018multi}, the authors claimed the existence of an online BST algorithm with cost $O((\log k)^7 \FF^k(X))$. However, this claim has an error since the algorithm employs Lee's \cite{lee2018fusible} $k$-server result, which the author has retracted. Instead of Lee's result, we can use the k-server result of Koutsoupias and Papadimitriou \cite{koutsoupias1995k}, $(2k-1)$-competitive. By using \cite{chalermsook2018multi}, it implies that an online BST algorithm exists with a running time of $O(k \FF^k(X))$.
This is a relatively large gap when compared with the best achievable offline BST bound of $O(\log k)F^k(X)$~\cite{chalermsook2018multi}; in fact, whether there exists a BST whose cost is additive in the $k$-finger bound, that is $O(F^k(X) + m \log k)$, has remained open. 


\paragraph{Unified Bound.} 
To describe the unified bound, we should first understand the {\em working set bound} \cite{iacono2005key,SleatorT85} and the {\em dynamic finger bound} \cite{cole2000dynamic,SleatorT85,IaconoL16}. If $x_t$ is a search key, then the working set bound requires $x_t$ to be searched in amortized time $O(\log \WS(x_t))$, where $\WS(x_t)$ is the number of distinct keys searched since the last search of $x_t$. The working set bound is based on temporal locality and implies many more bounds, such as the static finger bound,  the static optimality bound, etc. The dynamic finger bound states that the amortized time to search $x_t$ is $O(\log|x_t-x_{t-1}|+2)$. The dynamic finger bound is based on spatial locality.

The {\em unified bound} \cite{iacon_unified} implies both the working set and the dynamic finger. 
\begin{definition}
    The unified bound can be defined as $\UB(X) = \sum_{t=2}^m\log(\min_{t'<t}
\{t-t' + |x_t-x_{t'}| + 2\}).$
\end{definition} The unified bound is stronger than both the working set and the dynamic finger, as it suggests that it is in-expensive to search a key that is close to a recently searched key. It is true that the access lemma cannot prove the unified bound, although there are data structures that satisfy the unified bound. Iacono \cite{iacon_unified} gave a comparison-based data structure called the unified structure that achieves the unified bound. 
In \cite{buadoiu2007unified}, the authors gave a dynamic comparison-based data structure that achieves the unified bound. 
Derryberry and Sleator \cite{derryberry2009skip} designed the first BST algorithm called Skip-Splay tree that nearly achieves the unified bound with its running time of  $O(\UB(X) + m \log\log n)$.  \cite{bose2012layered} modified Skip-Splay using layered working-set trees to get amortized  $O(\UB(x_i) + \log\log n)$ time. Whether the unified bound can be achieved by a BST has remained an intriguing open question. 

\paragraph{Unified bound with time window.}  

In the unified bound, for each $x_t$, we find a key $x_{t'}$ that minimizes the term $
(t-t' + |x_t-x_{t'}| + 2)$ where $t'<t$. We can add another condition that $t'$ should be one of the last $k$ searched keys before time $t$. Thus, $t'$ comes from some time window. We call this variant {\em Unified Bound with a time window}. We formally define it as follows.

\begin{definition}
   Given an integer $k$, the unified bound with a time window is $$\UB^k(X) = \sum_{t=2}^m\log\Big(\min_{t'\in [t-k \dots t-1]}
\{t-t' + |x_t-x_{t'}| + 2\}\Big).$$
\end{definition}

This bound can be seen as ``interpolating'' between the dynamic finger and the unified bounds: 
When $k=1$, $\UB^k(X)$ is the dynamic finger bound, and when $k=m$, $\UB^m(X)$ is the unified bound. For $k=1$, we know that both Splay tree and $\GR$ satisfy the dynamic finger bound. The authors of  \cite{chalermsook2018multi}, showed a relation between $\UB^k(X)$ and $\OPT(X)$, i.e, $$\OPT(X) \le \beta(k) \UB^k(X)$$ where $\beta(\cdot)$ is some fixed (super exponentially growing) function.  

In sum, the access lemma is an intrinsic property of a BST which implies nice properties but cannot seem to imply any of the aforementioned properties. 

\subsection{Our New Concept: The Group Access Bounds} 

The main conceptual contribution of this paper is to introduce the \textit{group access bounds} that generalize the bound from the access lemma. We give an informal definition here (see \Cref{sec:group access tree bound} for the formal definitions).  Denote $[n] = \{1,\ldots, n\}$. The key gadget to describe our bound is the notion of \textit{group access tree}, which captures a hierarchical partition of $[n]$ until singletons are obtained. That is, in a group access tree $\tset$, each node $v \in V(\tset)$ is associated with an ``interval'' $I_v \subseteq [n]$ (consecutive integers). 
The root $r \in V(\tset)$ has $I_r = [n]$, if a node $v$ has children $v_1,\ldots, v_k$, then we have that $\{I_{v_1},\ldots, I_{v_k}\}$ forms a partition of $I_v$. This process continues until each leaf $v \in V(\tset)$ is associated with a singleton. 
Note that the tree does not have to be binary. See Figure~\ref{fig:group-tree} for illustration.

\begin{figure}[h]
    \centering
    \includegraphics[width=0.8\textwidth]{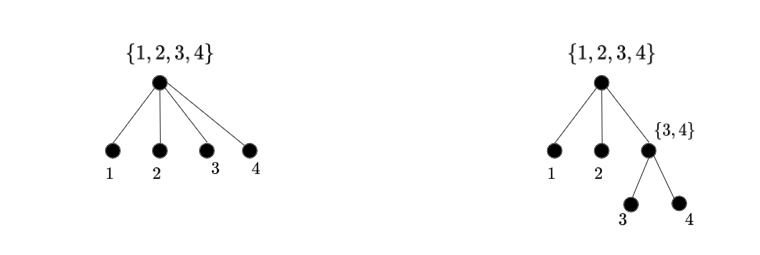}
    \caption{Examples of two group access trees. Each represents a hierarchical partition of $[n]$ until singletons are obtained. When a group access tree is a star (LHS), our bound is simply the access lemma. }
    \label{fig:group-tree}
\end{figure}

Let $w$ be a weight function that assigns a real-valued weight to each node in $\tset$. We define the cost to access the tree $\mathcal{T}$ w.r.t. the weight function $w$ as follows. For any edge $(u,v)$ where $u$ is the parent of $v$, the cost on $(u,v)$ is $ \log \frac{W(u)}{w(v)}$ where $W(u)$ is the total weight of all the children of $u$. The access cost of a key $a \in [n]$, denoted as $\text{cost}_{\mathcal{T},w}(a)$, is defined as the total cost of all the edges in the path $P_a$ from the root to the leaf containing $a$ in the group access tree $\mathcal{T}$. That is, 
$$\text{cost}_{\mathcal{T},w}(a) = \sum_{(u,v) \in P_a} \log \frac{W(u)}{w(v)}. $$

Some readers may have observed the similarity between our cost function and that of the access lemma. Indeed, one can show that it generalizes the access lemma. 
\begin{observation}
    If the group access tree $\mathcal{T}$ is a star, then the access $a \in [n]$ on $\mathcal{T}$ gives the same (amortized) cost as the access lemma on the same weight function. 
\end{observation}

Intuitively, the group access bound offers a ``search tree'' (which is not necessarily binary) where the cost of searching $a$ is the sum of the cost of the edges on the search path $P_a$.

Let  $\mathcal{W} = (w^{(1)}, \ldots, w^{(m)})$ be a sequence of weight functions. The total cost  of the group access tree $\mathcal{T}$ on an access sequence $X = (x_1,\ldots, x_m)$ , where $x_t \in [n]$ for all $t$, is 
$$ \text{cost}_{\mathcal{T}, \mathcal{W}}(X) = \sum_t\text{cost}_{\mathcal{T},w^{(t)}}(x_t). $$

Similar to the access lemma, the weight functions should change in a controllable manner in order to be meaningful (e.g., the weight functions for the working set bound are changed in a structured way in the access lemma).  Here, we introduce the notion of \textit{locally bounded} weight families. We say that a sequence of weight functions $\mathcal{W} = (w^{(1)},\ldots, w^{(m)})$ is \textit{locally bounded} if for all $t$, the weight increase from time $t$ to $t+1$ can happen only at the nodes in the path from root to $x_t$ in the group access tree $\mathcal{T}$. 

Given an input sequence $X = (x_1,\ldots,x_m)$,  \textit{the group access bound} $\gab(\tset, X)$ w.r.t.  a group access tree $\mathcal{T}$ is
\[ \gab(\tset,X) = \min_{\wset \mbox{ locally bounded}}  \text{cost}_{\mathcal{T}, \mathcal{W}}(X).  \]


\subsection{Our Technical Results: Deriving BST Bounds from GAB}

We show that the cost of the group access tree (with respect to certain weight functions) is competitive against many strong BST bounds that are not known via the access lemma.  
We say a group access bound is \emph{randomized} if the group access tree $\mathcal{T}$ is obtained by a random process that iteratively partitions $[n]$ until singletons are obtained.

\begin{theorem} \label{thm:main gab}
    For each of the following bounds, there exists deterministic group access trees $\mathcal{T}_1, \tset_3$ and randomized group access trees $\tset_2, \tset_4$ such that for all access sequences $X = (x_1,\ldots,x_m)$ where $x_t \in [n]$ for all $t$, 
\begin{enumerate}
    \item $\gab(\tset_1,X) = O(\sqrt{\log n}) \cdot \OPT(X)$ where  $\OPT(X)$ is the cost of offline optimal BST on $X$. 
    \item   $\gab(\tset_2,X) = O(\FF^k(X) + m \log k \log \log n)$ (randomized). That is, it is competitive with $k$-finger up to an additive term when the reference tree is an almost complete binary tree. 
    \item $\gab(\tset_3,X) = O(\UB(X) + m \log\log n)$. That is, it is competitive with the unified bound up to an additive term. 
    \item  $\gab(\tset_4,X) = O(\UB^k(X) + m \log k \log \log n)$ (randomized).  That is, it is competitive with unified bound with a time window up to an additive term. 
\end{enumerate}
The group access tree for the second and the fourth bound can be efficiently constructed from a probability distribution. 
\end{theorem}

It is not immediately clear that these group access bounds can be realized by BST algorithms. We show that every group access tree $\mathcal{T}$ can be simulated by an online BST algorithm. 
Let $\mathcal{A}$ be a BST algorithm. We denote $\text{cost}_{\mathcal{A}}(X)$ to be the cost of algorithm $\mathcal{A}$ running on a sequence $X = (x_1,\ldots, x_m)$. 
\begin{theorem} [Simulation Theorem]\label{thm:simulation}
For any group access tree $\mathcal{T}$, there exists an online BST algorithm $\mathcal{A}$ such that for any sufficiently long access sequence $X$ ,  $\text{cost}_{\mathcal{A}}(X) = O(\gab(\mathcal{T},X))$.  Furthermore, if the group access tree $\mathcal{T}$ is randomized, then $\mathcal{A}$ is randomized, where the competitive ratio is measured in the oblivious adversary model. 
\end{theorem}


The Simulation Theorem~(\Cref{thm:simulation}) signifies that one can prove new BST bounds by just proving the existence of a group access tree along with locally bounded weight families.

\begin{table}[h]
\centering
\begin{tabular}{|l|l|l|l|l|}
\hline 
                               & \textbf{Offline Upper}                    & \textbf{Offline Lower}       & \textbf{Online Upper}                     & \textbf{Our new bounds}                           \\
                               \hline 
 $\opt$  &  $\opt \log \log n$ \cite{tango_trees}& $\opt$ & $ \opt \log \log n$ \cite{tango_trees} & $\opt \sqrt{\log n}$ \\ 
$F^k$                       & $F^k(X) \log k$     \cite{chalermsook2018multi}           & $F^k(X) + m \log k$  & $k F^k(X)$      \cite{chalermsook2018multi}              & $F^k(X) + m \log k \log \log n$ \\
${\sf UB}$                  & ${\sf UB}(X) + m \log \log n$  \cite{derryberry2009skip}& ${\sf UB}(X)$       & ${\sf UB}(X) + m \log \log n$ \cite{derryberry2009skip} & ${\sf UB}(X) + m \log \log n$     \\
${\sf UB}^k$ & $\beta(k) {\sf UB}^k(X)$ \cite{chalermsook2018multi}        & ${\sf UB}^k(X)$     & $\beta(k) {\sf UB}^k(X)$  \cite{chalermsook2018multi}       & ${\sf UB}^k(X) + m \log k \log \log n$   \\
                               \hline 
\end{tabular}
\caption{Summary of the new bounds that can be derived via the group access bound. The first three columns show the best-known results prior to this paper. Asymptotic notations (Big-Oh and Big-Omega) are hidden for brevity. }
\end{table}

\subsection{Significance of our results}

Beyond the access lemma, group access bounds serve as the first step towards providing a unified framework for proving binary search tree bounds systematically. To resolve the dynamic optimality conjecture, a candidate algorithm must satisfy all dynamic optimality corollaries simultaneously, and we believe group access bounds are the starting point for this.

Apart from the competitiveness result, the results we present here either match the best-known bounds (such as unified bound) or provide an improvement upon even the best-known offline algorithms ($k$-finger bound when reference tree is an almost complete binary tree and unified bound with bounded time window) in the BST model. 

We remark that even though our competitive ratio does not match the $O(\log \log n)$ best-known factor, our algorithm (that we call $\GGR$) is very similar to $\GR$ -- a prime candidate for dynamic optimality. $\GGR$ resembles Greedy and inherits its conceptual simplicity. 
Even after a lot of work in this area \cite{fox11,IaconoL16,chalermsook2015greedy,chalermsook2015self,chalermsook2015pattern,chalermsook2018multi,chalermsook2020new,goyal2011dynamic,goyal2019better,derryberry2009skip,kozma2018smooth,chalermsook2023improved}, only the trivial bound for $\GR$ is known:  $\GR(X) = O(\log n)\ \OPT(X)$ for all possible $X$. Thus, $\GR$ is $O(\log n)$-competitive.
Our result raises some hope of proving $o(\log n)$-competitive ratio for Greedy by showing that Greedy satisfies the group access bound.

We illustrate the power of the group access bound by deriving two new BST bounds that have not been known for even offline BST algorithms. 

%
%
\begin{corollary} For each of the following items, there is a randomized online BST algorithm $\mathcal{A}$ such that the expected cost for any search sequence $X = (x_1,\ldots, x_m)$ where $x_t \in [n]$ for all $t$, 
\begin{itemize}
    \item  $\text{cost}_\mathcal{A}(X) = O(\FF^k(X) + m \log k \log \log n)$, and
    \item  $\text{cost}_\mathcal{A}(X) = O(\UB^k(X) + m \log k \log \log n)$. 
\end{itemize}
\end{corollary} 
The algorithms are randomized because the group access trees are chosen based on a probability distribution. 
The bound $O(F^k(X) + m \log k \log \log n)$ gives an improvement from the best known online BST that costs $O(k F^k(X))$ in~\cite{chalermsook2018multi} and improves upon the offline bound $O(\log k) F^k(X)$ for some range of parameter $k$ when the reference tree is an almost complete binary tree.

From such an improvement, we can derive a new ``pattern-avoiding'' bound for BSTs.  We say that a sequence \textit{contains} another sequence (or pattern) $\pi$ if it contains a subsequence that is order-isomorphic to $\pi$.

\begin{corollary}
Let $X \in [n]^m$ be a sequence that does not contain the pattern $(k,k-1,\ldots, 1)$. Then there exists a randomized BST that accesses $X$ with cost $O(nk + m \log k \log \log n)$.     
\end{corollary}

This bound improves the best-known online algorithm~\cite{chalermsook2023improved} that gives a bound of $O(m k^2)$ and the best-known bound on the offline optimum $O(mk)$~\cite{chalermsook2018multi} for some range of parameters when the reference tree used to calculate $F^k(X)$ is an almost complete binary tree. 

The unified bound obtained in this paper matches the best-known bound of $O(\UB(X) +m\log \log n)$ by Derryberry and Sleator \cite{derryberry2009skip}. In fact, we show that we can use the analysis of \cite{derryberry2009skip} as a black box once we obtain the group access tree for Unified bound. 

We do expect more applications of the group access bounds in binary search trees since group access bounds are generic and yet (unlike the dynamic optimality conjecture) maintain a certain flavor of being ``static'' (since $\tset$ is still fixed). We show that this static component can be algorithmically leveraged. 
We believe that our bound offers a ``bridge'' between the relatively static access lemma to the dynamic optimality conjecture, which requires a full understanding of dynamic BSTs.

\paragraph{Contribution to potential function analysis.} Another interesting aspect of our work is that once the group access bound is formulated, the proof relies solely on the use of the standard Sum-of-logs potential function, which in general, does not seem sufficient to prove any strong bounds beyond the access lemma. 
We show that by ``augmenting'' the access lemma with a group access tree, a natural and standard sum-of-logs potential function immediately provides significantly stronger BST bounds. 

In general, designing a potential function for analyzing a given algorithm is a highly innovative but rather ad-hoc task. Our work suggests that the sum-of-logs potential function on the group access tree might be a good candidate for proving that Greedy or Splay is $O(\sqrt{\log n})$-competitive.



\paragraph{Combining BSTs.} In \cite{demaine2013combining}, the authors showed that different BSTs with well-known bounds can be combined into a single BST that achieves all the properties of the combined BSTs. For example, Tango trees and Skip-Splay trees can be combined to get a BST which is $O(\log \log n)$-competitive and achieves the Unified bound with an additive term of $O(\log \log n)$. The combined BST from their approach, however, results in a different BST algorithm. 
Our group access bounds offer another way to combine known BST bounds, as we have illustrated in the above discussion, while retaining the original algorithm, e.g., proving that Splay satisfies the group access bound implies that Splay itself possesses all the nice properties derived from GAB. 
 


\subsection{Concluding Remarks}

We propose the group access bound --  a far-reaching extension of the standard access lemma and present applications in deriving new and unifying old bounds. Some of our bounds even improve the best-known upper bound on the offline optimum.

An immediate (and perhaps most interesting) open question is whether Greedy or Splay satisfies the group access bound via the sum-of-logs potential function. We believe that this question is very concrete (since it involves a specific potential function), so proving or refuting it would not be beyond the reach. 

Developing further understanding and finding more applications of our group access bounds are interesting directions. For instance, what are other BST bounds that can be implied by the group access bounds? Can we show that $\gab(\tset, X) \leq O(\opt(X))$ for some (distribution of) group access tree $\tset$?  Can one derive a non-trivial result about more general pattern-avoiding bounds~\cite{chalermsook2015pattern,goyal2019better}? Can we design a BST algorithm with a running time $O(\FF^k(X) + m\log k)$ for any sequence $X$ of length $m$ where the $k$-finger bound is calculated on any arbitrary reference tree?

There are also open questions to settle the complexity of specific BST bounds both in the online and offline settings. Most notably, is there any BST data structure that satisfies the unified bound?

\subsection{Organization}

We introduce notation and terminology in \Cref{sec:prelims}. We formalize the description of the group access bounds in \Cref{sec:group access tree bound}. We prove each of the items in \Cref{thm:main gab} in \Cref{sec:newsqrt,sec:newkfinger,sec:newunified,sec:newubtw}, respectively. To prove the Simulation theorem (\Cref{thm:simulation}), we first define an algorithm, called $\GGR$, that simulates the group access tree in \Cref{sec:ggreedyalgo}. We derive the group access lemma in \Cref{sec:newgroupaccess} and finally prove the Simulation theorem in \Cref{sec:newggreedycomp}.

\section{Preliminaries}
\label{sec:prelims}

\paragraph{The Geometric View}

Let $X=(x_1,x_2,\dots,x_m)$ be a sequence of $m$ accesses where each access is from the set $\{1,2,...,n\}$. This sequence can be represented as points in the plane, that is, $X_p = \{(x_t, t): t \in [m]\} \subseteq {\mathbb R}^2$. Imagine these points on a plane with an origin and X-Y axis. Since both the coordinates of a point 
are positive, all points lie in the first quadrant. The positive $x$-axis represents the key space, and the $y$-axis represents time. 
For any two points $p,q$ in a point-set $P$, if they are not in the same horizontal or vertical line, we can form a rectangle $\square pq$. 
A rectangle $\square pq$ is said to be {\em arborally satisfied} if $\exists r \in P \setminus \{p,q\}$ such that $r$ lies in $\square pq$.
\cite{demaine_geometry} introduced us to the following beautiful problem:

\begin{definition}(Arborally Satisfied Set)
	
\noindent Given a point set $X_p$, find a point set $Y$ such that  $|X_p \cup Y|$ is  minimum and  every pair of points in $X_p \cup Y$ is arborally satisfied.
\end{definition}
\cite{demaine_geometry} showed that finding the best BST execution for a sequence $X$ is equivalent to finding the minimum cardinality set $Y$ such that $X_p \cup Y$ is arborally satisfied. 

\begin{lemma}
\label{lem:bsttogeometry} (See Lemma 2.3 in \cite{demaine_geometry})
Let $A$ be an online algorithm that outputs an arborally satisfied set on any input representing $X$.  Then, there is an online BST algorithm $A'$ such that the cost of $A'$ is asymptotically equal to the cost of $A$, where the cost of $A$ is the number of points added by $A$ plus the size of $X$.
\end{lemma}

At time $t$, we say that $x_t$ is an {\em access key}. An algorithm adds (or touches) {\em points} while processing a key with the aim of making the final point set arborally satisfied. For a point $p$, denote $p.x$ and $p.y$ as its $x$-coordinate and $y$-coordinate, respectively. Note that $p.x$ denotes a key, and $p.y$ denotes the time when the point $p$ was added.

Let $q$ be a key. When we say that $\GR$ (or any other algorithm) adds a {\em point at key $q$} at time $t$, it means that $\GR$ adds a point at coordinates $(q,t)$.  Given two points $p$ and $q$, $q$ lies to the right of $p$ if $q.x > p.x$, else it lies to the left of $p$. For brevity, we will avoid using ceil and floor notation for various parameters used in this paper.

\section{The Group Access Bound} \label{sec:group access tree bound}



 In this section, we describe the \emph{group access bound}, which generalizes the access lemma. The concept of group access bound consists of two ingredients: 
\begin{itemize}
    \item \textbf{Hierarchical partition}: An \textbf{interval partition} of $[n]$ is a partition $\pi$ such that each set $S \in \pi$ is an interval (i.e. consecutive integers). Let $\pi$ be an interval partition of $[n]$. We say that $\pi$ is a refinement of another partition $\pi'$ if for all $S \in \pi$ and $S' \in \pi'$, we have $S \subseteq S'$ or $S \cap S' = \emptyset$. For instance $\{\{1,2\}, \{3,4\}, \{5,6\}\}$ is a refinement of $\{\{1,2,3,4\}, \{5,6\}\}$.  

    A hierarchical partition is a sequence of partitions $\pset = \{\pi_0, \pi_1, \ldots, \pi_k\}$ such that (i) for all $i$, $\pi_{i+1}$ is a refinement of $\pi_i$, (ii) $\pi_0 $= $\{[n]\}$ and (iii) $\pi_k = \{\{i\}\}_{i \in [n]}$ (singletons). 
    Given such $\pset$, each interval (set) in $\pi_i$ is referred to as a \textbf{group}. 
    The sets in $\pi_i$ are called level-$i$ groups for $\pset$. 

    A hierarchical partition $\pset$ has a natural corresponding tree ${\mathcal T}$ where each node in $V({\mathcal T})$ corresponds to a group. The root of ${\mathcal T}$ is $[n]$ (the group in $\pi_0$). Level-$i$ of the tree contains nodes that have 1-to-1 correspondence with sets in $\pi_i$. Moreover, there is an edge connecting $S \in \pi_i$ to $S' \in \pi_{i+1}$ if $S' \subseteq S$.   
    

    \item \textbf{Weight functions}: 
    Given a canonical hierarchical partition $\pset$, a $\pset$-weight function is an assignment of real values to nodes in $V(\tset)$, i.e., $w: V(\tset) \rightarrow {\mathbb R}_{\geq 0}$. 
    
\end{itemize}

We use the term \textbf{group access tree} to denote a hierarchical partition $\pset$ (as well as its corresponding tree $\tset$). 
See Figure~\ref{fig:supertree} for illustration. 

\begin{figure}[htp!]
  \centering

  \includegraphics[trim={0 140 10 60},clip,scale=.5]{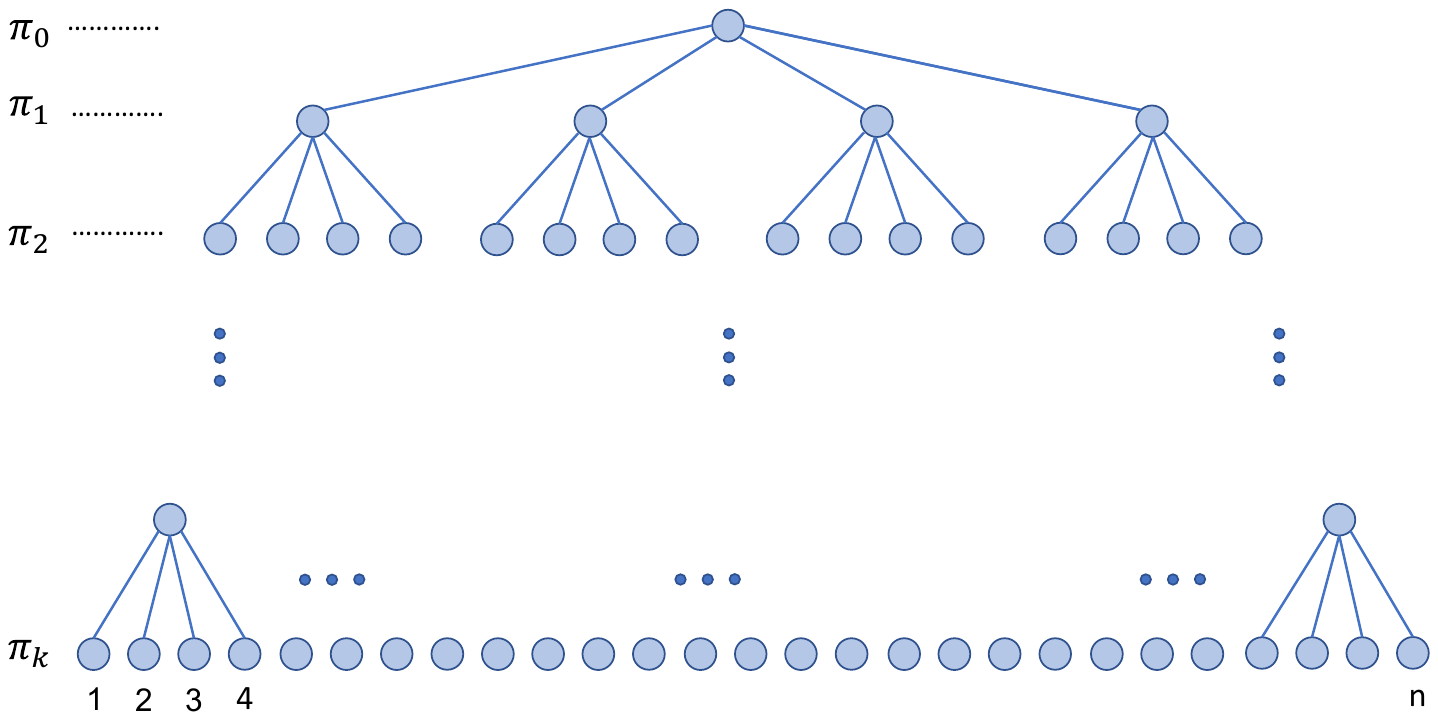}
  \caption{An illustration of group access tree} 
  \label{fig:supertree}
\end{figure}






\begin{definition} (\textbf{Group given a key})
\noindent	Let $\pset$ (or $\tset$) be a group access tree and  let $p \in [n]$ be a key. Then, for each $j$, we use $g_j(p)$ to denote the (unique) level-$j$ group $S \in \pi_j$ in which $p$ lies.  
\end{definition}



\begin{definition}(\textbf{Level $j$ groups of a group})
Given a level-$(j-1)$ group (interval) $g$,  denote by $\dg({g})$ the set of children groups in level $j$ that are contained in $g$, i.e., $\dg(g)= \{g' \in \pi_{j}: g' \subseteq g\}$. These are exactly the same as the groups that are children of node $g \in V(\tset)$. 
\end{definition}

We are now ready to define the access cost in the group access tree.
When accessing key $a \in [n]$ in the group access tree ${\mathcal T}$ and $\pset$-weight $w$, denote by $\tset(a)$ the path from root to $a$ in the tree $\tset$; in particular, this path contains $\tset(a) = (g_0(a), g_1(a), \ldots, g_k(a) = \{a\})$.  The cost incurred on edge $e = (g_{j-1}(a), g_{j}(a))$ is 
\[c_e({\mathcal T}, w, a) = \log \left(\frac{W^j}{w(g_j(a))}\right)\]

where $W^j = \sum_{g' \in \dg(g_{j-1}(a))} w(g')$. 
The total access cost is $c({\mathcal T}, w, a) = \sum_{e \in \tset(a)} c_e({\mathcal T}, w, a)$. 
Notice that this access cost is very similar to the access lemma cost. In fact, one can show that it generalizes the access lemma:  

\begin{observation}
The access lemma corresponds to the cost $c(\tset,w, a)$  when $\tset$ is a star. 
\end{observation}

In this way, our group access bound can be seen as an attempt to strengthen the standard access lemma by introducing hierarchical partitioning. As in the access lemma, the interesting application to BSTs happens when the changes of weights are ``controllable''  (e.g. in the working set bound). We introduce the concept of \emph{locally bounded} weight families $\wset$ to capture this property. 

\begin{definition}
The weight family $\wset$ is locally bounded if 
for all time $t$, for every group $g \not \in \tset(x_t)$, we have $w_{t+1}(g) \leq w_t(g)$.  This means that the weight can increase only when $g \in \tset(x_t)$. 
\end{definition}

We are interested in the total access cost on a sequence $X = (x_1,\ldots, x_m) \in [n]^m$ where the group access trees are allowed weight changes over time, that is, we are given a sequence of weight functions $\wset = \{w_1,\ldots, w_m\}$ where $w_t$ denotes the weight function at time $t$.

The \textbf{group access bound} w.r.t. $(\tset,X)$ is: 
\[ \gab(\tset,X) = \min_{\wset \mbox{ locally bounded}} \sum_{t \in [m]} c(\tset, w_t, x_t) \]

Note that the use of "minimum" in the definition of group access bound serves to select the weight family with the lowest weight among multiple locally bounded weight families.

Our main contribution is in showing that the group access bound is competitive to many strong bounds in the binary search tree model.
We say that the group access bound is $(\alpha,\beta)$-\textbf{competitive} to function $f: [n]^m \rightarrow {\mathbb R}$ if there exists a group access tree $\tset$ such that $\gab(\tset,X) \leq \alpha f(X) + \beta |X|$.

\begin{theorem}\label{thm:treesqrt} 
The group access bound is $(O(\sqrt{\log n}), O(1))$-competitive to ${\sf OPT}$.  
\end{theorem}

\newcommand{\dset}{{\mathcal D}}

The group access bound can also be used by allowing randomization in choosing the hierarchical partition $\pset$. We say that the group access bound is randomized  $(\alpha,\beta)$-\textbf{competitive} to function $f: [n]^m \rightarrow {\mathbb R}$ if there is an efficiently computable distribution $\dset$ that samples $\tset$ such that 
\[ {\mathbb E}_{\tset \sim \dset}[\gab(\tset,X)] \leq \alpha f(X) + \beta |X|\]
 
\begin{theorem}\label{thm:treekfinger}
The group access bound is (randomized) $(O(1), O(\log k \log \log n))$-competitive to the $k$-finger bound when the reference tree is an almost complete binary tree. 
\end{theorem}

\begin{theorem}\label{thm:treeub}
The group access bound is $(O(1), O(\log \log n))$-competitive to the unified bound. 
\end{theorem}

\begin{theorem}\label{thm:treeubtw}
The group access bound is (randomized) $(O(1), O(\log k \log \log n))$-competitive to the unified bound with time window of size $k$. 
\end{theorem}

We say that a BST algorithm $\aset$ \textbf{satisfies the group access bound} w.r.t. group access tree $\tset$ if the cost of the algorithm is at most $O(\gab(\tset, X))$ for all input sequence $X$. 

\begin{proposition}
If a BST algorithm satisfies the group access bound and the group access bound is  $(\alpha,\beta)$-competitive to function $f$, then the cost of the BST algorithm on any sequence $X$ is at most $O(\alpha \cdot f(X) + \beta|X|)$.     
\end{proposition}

We will later show that a family of BST algorithms (named $\GGR$) satisfies the group access bound and possesses all the aforementioned properties.  
\section{$(O(\sqrt{\log n}),O(1))$-competitiveness}\label{sec:newsqrt}

In this section, we prove Theorem~\ref{thm:treesqrt}. We will define the appropriate group access tree so that the group access bound is upper bounded by $(O(\sqrt{\log n}), O(1))\cdot \opt(X)$. 

\subsubsection*{Group access tree} 
Define the partition $\pset$ inductively as follows. Let $M = 2^{\sqrt{\log n}}$. First $\pi_0 = \{[n]\}$. Given $\pi_i$, we define $\pi_{i+1}$ by, for each interval $S \in \pi_i$, partitioning $S$ into $S_1, S_2, \ldots, S_M$ equal-sized intervals and adding them into $\pi_{i+1}$. This would give us a group access tree where each non-leaf node has   $M$ children and its height is at most $h=O(\sqrt{\log n})$.

\subsubsection*{Weight function} 
Given a sequence $X= (x_1,x_2,\ldots, x_m)$, we define a weight function $\wset$ which is locally bounded and such that $c(\tset, \wset, X) = \sum_t c(\tset, \wset, x_t) \leq O(\sqrt{\log n})\cdot (\opt(X)+m)$. This will give us the desired result. 
Our weight function uses the notion of \emph{last access}. 

\begin{definition}
Consider time $t$ and the group $g = g_{j-1}(x_t)$. Let $t' <t$ be the last time before $t$ at which $g$ is on the search path $\tset(x_t)$. We say that a child $g_1 \in \dg(g)$ is \textbf{last accessed (child) group} of $g$ at time $t$ if the edge $(g,g_1)$ is on search path $\tset(x_{t'})$.     
\end{definition}

Remark that each group can have at most one child in $\tset$ that is the last access group. Now, we are ready to define the weight function: 

$w_t(g) = 
\begin{cases}\label{eq:weightsqrt}
	M  \hspace{1.3cm} \text{if $g$ is the last accessed group of its parent}\\
        1\hspace{1.5cm} \text{otherwise}
\end{cases}$

It is easy to verify that this family of weight functions $\wset$ is locally bounded (there is only one key whose weight can increase between time $t$ and $(t+1)$). 

\begin{lemma}
Consider edge $e = (g_{j-1}(x_t), g_j(x_t))$. We have 
   $$c_e(\tset, \wset,x_t)=
\begin{cases}
	O(\log M)  \hspace{1.44cm} \text{if $g_j(x_t)$ is not the last accessed group of $g_{j-1}(x_t)$}\\
        O(1)\hspace{2.2cm} \text{otherwise}
\end{cases}$$
\end{lemma}
\begin{proof} 
    By definition, $W^j =\sum_{g' \in \dg(g_{j-1}(x_t))} w_t(g') \leq 2M$ because there is only one group in $\dg(g_{j-1}(x_t))$ with weight $M$ (the last accessed group) and there can be at most $M$ groups with weight 1.

    We analyze the cost in two cases:
    If $g=g_j(x_t)$ is the last accessed group, then we have $c_e(\tset, \wset, x_t) \leq \log (W^j/w_t(g)) \leq \log (2M/M)=O(1)$. 
    Otherwise, if $g_j(x_t)$ is not the last accessed group, then we have $\log (2M) \leq O(\log M)$.  
\end{proof}

\subsubsection*{Cost analysis}
Consider the search path $\tset(x_t)$. Denote by $\gamma_t$ the number of groups on $\tset(x_t)$ that are not the last accessed group of its parents. 
The cost $c(\tset, w_t, x_t) = \sum_{e \in \tset(x_t)} c_e(\tset, w_t,x_t)$, and from the above lemma, the cost is at most $O(\log M) \cdot \gamma_t + O(h) \leq O(\sqrt{\log n}) (\gamma_t+1)$. 
Therefore, the total cost is $O(\sqrt{\log n}) \cdot (\sum_t \gamma_t + m)$. 
The sum of $\gamma_t$ will be upper bounded by the Wilber bound. 

\paragraph{The Wilber bound:} Wilber \cite{wilber1989lower} gave two lower bounds on the running time of any BST on a sequence $X$. These bounds are known as $\WBO$ and $\WBT$.\footnote{They can also be derived using an elegant geometric language  of~\cite{demaine_geometry}.} We now describe $\WBO(X)$. 
Let $R$ be a leaf-oriented (keys at the leafs) binary search tree, and for each $a \in [n]$, denote by $R(a)$ the search path in $R$ of key $a$.  
When searching a sequence $X$ in $R$, for each node $v \in V(R)$, the \textbf{preferred child} of node $v$ at time $t$ (denoted by $\PR_t(v)$) is the child of $v$ on the last search path in $R$ at time $t$.  
If node $v$ is not on the search path $R(x_t)$, we know that the preferred child cannot change, i.e., $\PR_t(v) = \PR_{t-1}(v)$. 

The Wilber bound with respect to $R$ at time $t$ and node $v$ is: 
$$\WBO^t_R(v) =
\begin{cases}
	1 \hspace{1cm} \text{if} \ \PR_t(v) \neq \PR_{t-1}(v)\\
	0 \hspace{1cm} \text{otherwise}
\end{cases}
$$
The total Wilber bound of a sequence $X$ is $\WBO_R(X) = \sum_{t} \sum_{v} \WBO_R^t(v)$. The Wilber bound is defined as the maximum, over all reference BST $R$, of $\WBO_R(X)$. 

\begin{lemma}
We have that $\sum_{t} \gamma_t \leq O(\WBO(X))$    
\end{lemma}
\begin{proof}
It is sufficient to define a reference tree $R$ that allows us to charge the cost of $\sum_t \gamma_t$ to $\WBO_R(X)$. Notice that our group access tree $\tset$ is not a binary search tree. However, it can be naturally extended into a binary search tree $R$ as follows: We process the non-leaf nodes in $V(\tset)$ in non-decreasing order of distance from the root. When a group $g \in V(\tset)$ is processed, we remove the edges from $g$ to its children. Let $T_g$ be an arbitrary BST rooted at $g$ and leafs $\dg(g)$; we add the tree $T_g$ in place of the deleted edges. After all vertices are processed, it is straightforward to see that  the resulting tree is a (leaf-oriented) BST. 

Now we claim that $\sum_t \gamma_t$ can be upper bounded by $\WBO_R(X)$. Recall that $\gamma_t$ is the number of groups on the search path $\tset(x_t)$ that are not last accessed. Let $G_t \subseteq V(\tset(x_t))$ be those groups on the search path that are not last accessed. For each such group $g_1 \in G_t$, let $g$ be its parent, so we know that the last time $g$ was on the search path, some other group $g_2 \in \dg(g)$ was instead chosen. Notice that the search paths $R(g_1)$ and $R(g_2)$ also visit $g$ but branch away at some vertex $b$ inside $T_g$ (it could be that $b= g$). This would imply that $\PR_t(b) \neq \PR_{t-1}(b)$ and therefore $\WBO^t_R(b) =1$. This implies that $\gamma_t \leq \sum_{v} \WBO^t_R(v)$ and hence the lemma. 
\end{proof}

The lemma implies that the total group access bound is at most $$O(\sqrt{\log n}) (\WBO(X)+m)\leq O(\sqrt{\log n})\cdot  (\opt(X)+m)$$

\begin{figure}[htp!]
  \centering

  \includegraphics[trim={0 200 10 100},clip,scale=.5]{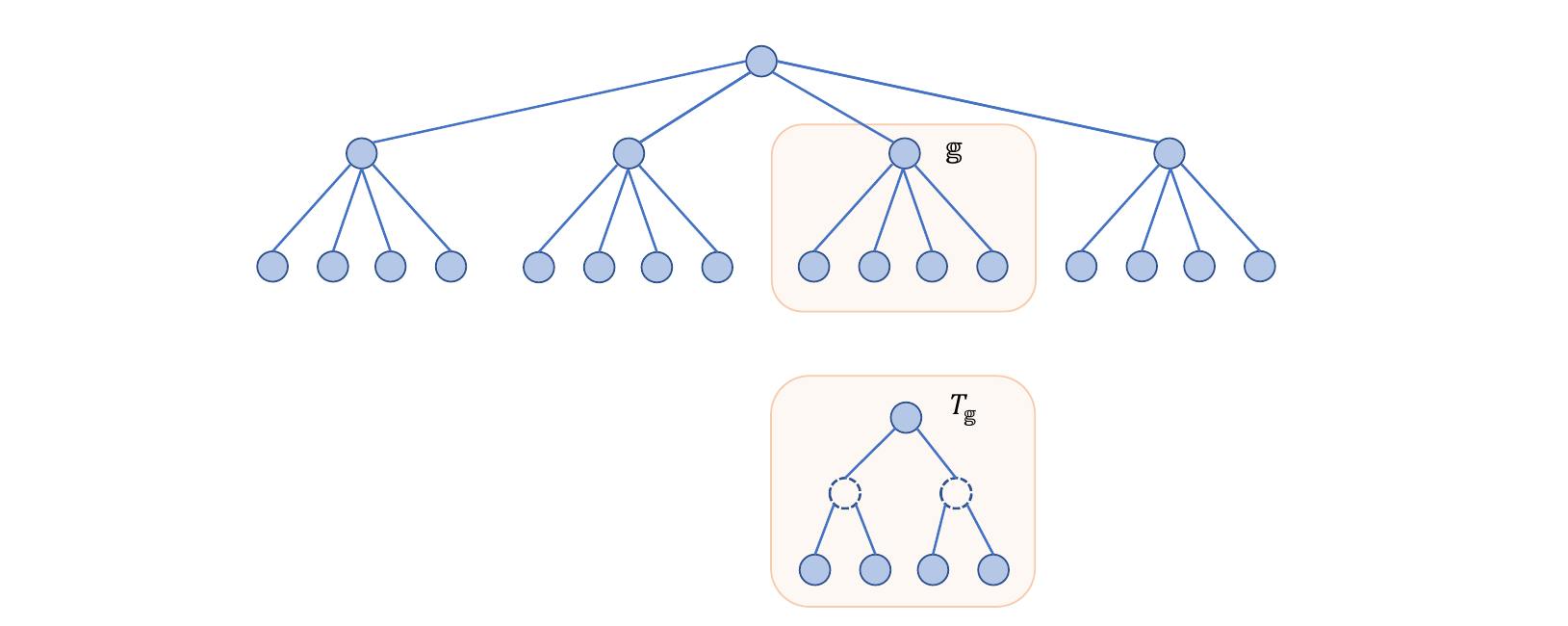}
  \caption{An example of $g$ and $T_g$} 
  \label{fig:supertree}
\end{figure}

\section{The $k$-finger bound}\label{sec:newkfinger}
In this section, we prove Theorem~\ref{thm:treekfinger}. We first define the appropriate group access tree such that the group access bound is $((O(1), O(\log k \log \log n))$-competitive with $k$-finger.

Define the partition $\pset$ inductively as follows. Let $\EL_i = \frac{\log n}{2^i}$. First $\pi_0 = \{[n]\}$. Given $\pi_i$, we define $\pi_{i+1}$ by, for each interval $S \in \pi_i$, partitioning $S$ into $S_1, S_2, \ldots, S_{\EL_i}$, almost equal-sized intervals and add them into $\pi_{i+1}$. This builds a group access tree of height at most $h=O(\log \log n)$. However, unlike the previous section, we will induce some randomization while constructing $\pi_{i+1}$ from $\pi_i$.

Let $g$ be a group in $\pi_{j-1}$. Let the interval of $g$ contain keys $(a,a+1,\dots,b)$. Pick an integer  $s\in [2^{\EL_j}]$ uniformly at random. Partition $(a,a+1,\dots ,b)$ into groups (intervals) $g_1 = (a,\dots,a+s-1), g_2 = (a+s,\dots,a+s+2^{\EL_j})$, $ g_3 = (a+s+2^{\EL_j}+1, \dots, a+s+2^{\EL_j+1})$ and so on. Thus, each group contains $2^{\EL_j}$ keys except possibly the first and the last one. $\pi_{j}$ contain groups $g_1, g_2, g_3, \dots $ and so on.

 Some remarks are in order. There are $O\big(2^{\EL_1}\big) = O(\sqrt n)$ groups at level 1, and each level 1 group is of size $\frac{n}{2^{\EL_1}} \le \sqrt n$. If $g$ is a level $j-1$ group, then $|\dg(g)|=O(2^{\EL_j}) = O(n^{1/2^j})$  . Using induction assume that there are $\le n^{1/2^{j-1}}$ keys in $g$. Then the number of keys in $g_1\in \dg(g)$  is $\le \frac{n^{1/2^{j-1}}}{2^{\EL_j}} = n^{1/2^j}$.  Thus, we observe the following:

 \begin{observation}
\label{obs:treenoofleveljgroup}
	Let $g$ be a level $j-1$ group. Then $|\dg(g)|=O(n^{1/2^j})$. Each group in $\dg(g)$ represents a set of consecutive keys of size $\le  n^{1/2^j}$.
\end{observation}

Let $\CC$ be an almost complete binary tree in which all the leaf nodes contain keys. $\CC$ is denoted as the reference tree. Let us assume that there are $k$ fingers stationed at some leaves of $\CC$. When $x_t$ is searched at time $t$, the finger that moves to $x_t$ in $\CC$ will be denoted by $f^k(t)$. The finger $f^k(t)$ is moved to $x_t$ via the shortest path in $\CC$. Define $\FF^k(t)$ as the length of this shortest path plus 1. The $k$-finger bound, $\FF^k(X)$, is defined as $\FF^k(X) = \sum_{t=1}^m \FF^k(t)$.

We are now ready to define the weight of a group. Let $g_1$ be a level $j$ group, then weight of $g_1$ at time $t$ is:

$w_t(g_1) = 
\begin{cases}\label{eq:weightsqrt}
	\max\{\frac{1}{k}, \frac{1}{n^{1/2^j}}\} \hspace{1.2cm} \text{if group $g_1$ contains a finger }\\
        \frac{1}{n^{1/2^j}}\hspace{2.6cm} \text{otherwise}
\end{cases}$

Note that the weight function is locally bounded from time $t-1$ to $t$. Consider the edge $e = (g_{j-1}(x_t), g_j(x_t))$, we will abuse the notation and denote $e_j=(g_{j-1}(x_t),g_j(x_t))$. We now bound the cost  $c_{e_j}(\tset, \wset,x_t)$ when $k < n^{\frac{1}{2^j}}$. The same analysis works when $k\ge n^{\frac{1}{2^j}} $. Note that when  $k < n^{\frac{1}{2^j}}$, $w_t(g_1)=\frac{1}{k}$ if $g_1$ contains a finger.

Let $W^j=\sum_{g'\in \dg(g_{j-1}(x_t))}w_t(g')$, then $W^j\le c$. Let $x_t$ be a key accessed at time $t$ and let $g=g_{j-1}(x_t)$, $g_1=g_{j}(x_t)$. Then we prove the following lemma:

\begin{lemma}
\label{lem:newcostkfinger}
Assume that $k < n^{1/2^j}$. The  cost $c_{e_j}(\tset, \wset,x_t)$ at time $t$ is:

 $c_{e_j}(\tset, \wset,x_t) =
\begin{cases}
	O(\frac{1}{2^j} \log n) \hspace{0.8cm}\text{if no finger lies in  $g_1$ at time $t-1$} \\
	O(\log k) \hspace{1.2cm} \text{otherwise}
\end{cases}$
\end{lemma}
\begin{proof}
    Let $W^j=\sum_{g' \in \dg(g_{j-1}(x_t))} w_t(g')$, then $W^j\le c$.

    If $g_1$ has a finger at time $t-1$, then its weight is $\frac{1}{k}$. Therefore,
    $$c_{e_j}(\tset, \wset,x_t)=\log\left(\frac{c}{1/k}\right)=O(\log k)$$
    Otherwise, if $g_1$ has no finger at time $t-1$, then its weight is $\frac{1}{n^{1/2^j}}$. Therefore,
    $$c_{e_j}(\tset, \wset,x_t))=\log\left(\frac{c}{1/n^{1/2^j}}\right)=O(\frac{1}{2^j}\log n)$$
\end{proof}

\subsection{Finishing the proof of \Cref{thm:treekfinger}}
In the ensuing discussion, we will show that the cost of the group access tree in \Cref{lem:newcostkfinger} can be {\em paid} by the $k$-finger bound with some additive term. To this end, we first prove the following useful lemma.  This lemma states that if  $\FF^k(t)$ is sufficiently small,  then $f^k(t)$ and $x_t$ will lie in the different level $j$ group with a small probability.

\begin{lemma}
\label{lem:newsamegroupprobability}
	Let $j$ be the largest integer such that  $\FF^k(t) \le \log\Big(\frac{n^{1/2^j}}{\log n^{1/2^j}}\Big)$. Then, for any $i \in [0,j]$, $$Pr\Big[ \text{$x_t$ and  $f^k(t)$ lie in different level $i$ group}  \Big] \le \frac{cd}{\log n^{1/2^i}}$$  for some constant $c \ge 1$, $d \ge 1$.
\end{lemma}
\begin{proof}
    Let $Y_p$ be the event that $x_t$ and $f^k(t)$ lie in different level $p$ groups but the same $p-1$ level group (where $1 \le p \le i$). Then the probability that $x_t$ and $f^k(t)$ lie in the different level $i$ group is:  $Y =  \cup_{p=1}^i Y_p$. We now calculate $Y_p$.
    
	Let us assume that $x_t$ and $f^k(t)$ lie in the same level $p-1$ group.  Let $g_{p-1}(x_t)$ contains the keys $(a, a+1, \dots , b)$. When we make groups at level $p$ in $g_{p-1}(x_t)$, we choose a random number $s \in [2^{\EL_p}] =  [n^{1/2^p}] $ and partition the keys using $s$. Since $\FF^k(t) \le \log \Big(\frac{n^{1/2^j}}{\log n^{1/2^j}}\Big)$, there are at most $\frac{c n^{1/2^j}}{\log n^{1/2^j}}$ keys between $s$ and $f^k(t)$ in the reference tree $\CC$ for some constant $c$. Thus, there are at most $\frac{cn^{1/2^j}}{\log n^{1/2^j}}$ values of $s$ for which $x_t$ and $f^k(t)$ lie in different groups. This implies the following:
	
	\begin{align*}
		\Pr[Y_p] &= \frac{cn^{1/2^j}}{\log (n^{1/2^j}) \times n^{1/2^p}} \Pr[\text{$x_t$ and $f^k(t)$ lie in the same level $p-1$ group}]\\
		& \le \frac{cn^{1/2^j}}{\log (n^{1/2^j}) \times n^{1/2^p}}
	\end{align*}
	
	We can now bound $Y$ as follows:
	\begin{align*}
		\Pr[Y] &= \sum_{p=1}^i \Pr[Y_p]\\
		&= 	\frac{cn^{1/2^j}}{\log (n^{1/2^j})} \sum_{p=1}^{i}\frac{1}{n^{1/2^p}}\\
		&\le \frac{cn^{1/2^j}}{\log (n^{1/2^j})} \frac{d}{n^{1/2^i}}
	\end{align*}
For all $i \le j$,  the above quantity is $\le \frac{cd}{\log n^{1/2^i}}$. This completes the proof.
	
\end{proof}

Let $j$ be the largest integer such that  $\FF^k(t) \le \log\Big(\frac{n^{1/2^j}}{\log n^{1/2^j}}\Big)$. Thus, $\FF^k(t) > \log \Big(\frac{n^{1/2^{j+1}}}{\log n^{1/2^{j+1}}}\Big)$ and certainly $\FF^k(t) > \log(n^{1/2^{j+2}})$. Hence, the number of keys between $s$ and $f^k(t)$ is $> cn^{1/2^{j+2}}$ where $c$ is a constant. At any level $\ge j+2$, we claim that $x_t$ and $f^k(t)$ will lie in different groups. This is because, using \Cref{obs:treenoofleveljgroup}, the group size is less than the number of keys between $x_t$ and $f^k(t)$. Thus, at all levels, $\ge j+2$, the group containing $x_t$ does not contain any finger.  Using \Cref{lem:newcostkfinger}, the cost of $c_{e_i}(\tset, \wset,x_t))$ from level $j+2$ to $\log \log n$ is 

\begin{align*}
	\sum_{i=j+2}^{\log \log n} c_{e_i}(\tset, \wset,x_t)) &= \sum_{i=j+2}^{\log \log n} \frac{1}{2^i} \log n \\
	&= O\Big(\frac{1}{2^j} \log n\Big) \\
	&= O(\FF^k(t)) \numberthis \label{eq:jton}
\end{align*}

For level $j+1$, using  \Cref{lem:newcostkfinger}, 
\begin{equation}
\label{eq:j+1}
	c_{e_{j+1}}(\tset, \wset,x_t)) = O(\max\{\frac{1}{2^{j+1}} \log n, \log k\}) = O(\max\{\FF^k(t), \log k\})
\end{equation}
For the levels $i \in [0,j]$, $x_t$ and $f^k(t)$ may lie in the same level $i$ group. The expected cost of $c_{e_i}(\tset, \wset,x_t))$ in iteration $i \in [0,j]$ can be given using \Cref{lem:newsamegroupprobability} and \Cref{lem:newcostkfinger}:

\begin{align*}
E[c_{e_i}(\tset, \wset,x_t))] &= \Pr[x_t \ \text{and $f^k(t)$ lie different level $i$ group} ] \times \frac{1}{2^i} \log n \\
& \hspace{1cm}+ \Pr[x_t \ \text{and $f^k(t)$ lie in same level $i$ group} ] \times \log k\\
\intertext{Using \Cref{lem:newsamegroupprobability}, we get}
&\le \frac{c^2}{\log n^{1/2^i}} \frac{1}{2^i} \log n +  \Big(1 - \frac{1}{\log n^{1/2^i}}\Big) \log k\\
&\le c^2 + \log k\\
&= O(\log k) \numberthis \label{eq:0toj}
\end{align*}

Thus, the total expected cost of $c(\tset, \wset,x_t))$ at time $t$ is 

\begin{align*}
\sum_{i=1}^{\log \log n} E[c_{e_i}(\tset, \wset,x_t))] &= \sum_{i=1}^{j} E[c_{e_i}(\tset, \wset,x_t))] + E[c_{e_{j+1}}(\tset, \wset,x_t))] + \sum_{i=j+2}^{\log \log n} E[c_{e_i}(\tset, \wset,x_t))] \\
\intertext{Using \Cref{eq:jton}, \Cref{eq:j+1}, and \Cref{eq:0toj}, we get }\\
&\le (j+1) \log k + O(\max\{\FF^k(t), \log k\})+ O(\FF^k(t))\\
&\le O(\log \log n \log k) + O(\FF^k(t))
\end{align*}

This implies that the expected cost of $c(\tset,\wset,X)$ is: 

\begin{flalign*}
	E[c(\tset,\wset,X)] &= \sum_{t=1}^m\sum_{i=1}^{\log \log n} E[c_{e_i}(\tset,\wset,x_t))]\\
	& \le \sum_{t=1}^m O(\FF^k(t) + \log k \log \log n)\\ 
	&= O(\FF^k(X) + m \log k \log \log n)
\end{flalign*}

Note that the above cost of $c(\tset,\wset,X)$ is expected. But this implies that there exists a group access tree such that its cost on the sequence $X$ is $(O(1), O(\log k \log \log n))$-competitive with $k$-finger when the reference tree is an almost complete binary tree.  Thus, we have proven \Cref{thm:treekfinger}.

\section{Unified Bound}\label{sec:newunified}
In this section, we will show that the group access bound is $O(\UB(X) + m\log \log n)$ for any sequence $X$. This bound matches the best-known bound for the Unified bound in the BST model (Skip-Splay tree) by Derryberry and Sleator \cite{derryberry2009skip}. In fact, we show that structurally the group access tree for Unified bound and Skip-Splay tree is the same. To show this correspondence, we will provide some notations and definitions from \cite{derryberry2009skip} and relate them with the definitions we have used in this paper.

The Skip-Splay tree $T$ is initially perfectly balanced with keys $\{1,\dots,n\}$ where $n=2^{2^{k-1}}-1$ for some $k>0$. The height of the tree $T$ is $\log \log n$, where the size (number of keys) of each internal node at any level is the square of the size of the node one level deeper, and the leaf nodes have a constant number of keys. Note that the tree $T$ is the same as the group access tree defined for $k$-finger if we omit the randomization part. 

We now define some notations from \cite{derryberry2009skip} and provide similarities with our definitions.
\begin{enumerate}
    \item Let $\rho_k=1$ and for $i<k$ let $\rho_i=2^{2^{k-i-1}}$ such that $\rho_i=\rho_{i+1}^2$. 
    According to our definition, $\rho_i$ is nothing but $|\pi_i|$ for the group access tree $\mathcal{T}$.

    \item Let $R_i(x)$ be the level-$i$ region of $x \in  T$, which is defined as follows. First, define the offset $\delta_i=\delta \text{ mod } \rho_i$, where $\delta$ is an integer that is arbitrary but fixed for all levels of $T$. Then, $R_i(x)=R_i^*(x)\cap T$ where

    $$R_i^*(x)=\left\{\left\lfloor\frac{ x+\delta_i }{\rho_i}\right\rfloor\rho_i - \delta_i,\dots, \left\lfloor\frac{ x+\delta_i }{\rho_i}\right\rfloor\rho_i - \delta_i + \rho_i -1\right\}.$$

    With the definition of $R_i(x)$ in mind we define $g_i(x)$ the same as $R_i(x)$. Therefore, the group access tree is not exactly the same as defined for $k$-finger, where we used randomization to partition the intervals. Instead, here we use the definition of $R_i(x)$ to partition the intervals (groups).

    \item Let $\mathcal{R}_i(x)$ be the set of level-$i$ regions that are subsets of $R_{i-1}(x)$. According to our definition $\mathcal{R}_i(x)$ is nothing but $\dg(g)$ where $g=g_{i-1}(x)$.

    \item For $x\in T$, let $w_i(x,t)$ represent the number of regions in $\mathcal{R}_i(x)$ that contains a query (access) since the previous access to a member of $R_i(x)$ at time $t$. Note that $w_i(x,t)$ is the working set number of the level-$i$ region $R_i(x)$ with respect to $\mathcal{R}_i(x)$. We now provide a similar definition for the working set of a level $i$ group.

    \begin{definition}(\textbf{Working Set of a level $i$ group}):
        Let $g_1$ be a level $i$ group such that $g_1 \in \dg(g)$ where $g=g_{i-1}(x_t)$.
 Let $t'$ be the largest integer less than $t$ such that $x_{t'}$ lies in $g_1$. Then, the working set number of $g_1$ at time $t$ is  $\WS_t(g_1) = t-t'$.
\end{definition}
 By definition, there can be at most one level $i$ group in $\dg(g)$ with working set number $1,2,3,\dots$, and so on. Let us define the weight of $g_1$ at time $t$ as :
 $w_t(g_1) = \frac{1}{(\WS_t(g_1))^2}$. Note that this weight function is locally bounded from time $t-1$ to $t$.
 \end{enumerate}

 The analysis, which proves that the Skip-Splay tree is $O(\UB(X) + m\log \log n)$ for a sequence $X$ of length $m$, consists of three lemmas. The first lemma divides the cost of the Skip-Splay tree into "local working set cost" with one cost term per level in $T$. We will also show that the cost of the group access tree $\mathcal{T}$  can be divided into local working set costs. Therefore, the group access tree $\mathcal{T}$ for the Unified bound also satisfies Lemma 1 of \cite{derryberry2009skip}. 
 
 The second lemma, which uses the first lemma, shows that the Skip-Splay tree satisfies a bound similar to the Unified bound plus an additive $O(\log \log n)$. The proof of this lemma only uses the structure of the Skip-Splay tree $T$ and the first lemma. As the structure of the group access tree $\mathcal{T}$ for the unified bound is the same as $T$, we immediately satisfy Lemma 2 of \cite{derryberry2009skip}. 

 The third lemma of \cite{derryberry2009skip} shows that the bound proved in Lemma 2 is within a constant factor of the Unified bound, which gives the desired result. For completeness, we provide the three lemmas derived in  \cite{derryberry2009skip} and then prove that Lemma 1 of \cite{derryberry2009skip} is satisfied by the group access tree $\mathcal{T}$ for the unified bound.

\begin{lemma}(Lemma 1 of \cite{derryberry2009skip})
    For a query sequence $X$ that is served by a Skip-Splay tree $T$ with $k$-levels, the amortized cost of the query $x_j$ using an arbitrary value of $\delta$ to define the regions is 

    $$O\left(k+ \sum_{i=1}^k \log w_i(x_j,j)\right).$$
\end{lemma} 

Before describing the second lemma, we first define the bound, which is very similar to the Unified bound as described in \cite{derryberry2009skip}. Let $ub(j) = \text{argmin}_{j'<j}\{j-j' + |x_j - x_{j'}| +2\}$ and let $t(x_j,j)$ represent the number of queries (instead of distinct keys accessed) since the previous access to $x_j$. Define $\UB'(X)$ a variant of Unified bound as follows

$$\UB'(X)=\sum_{j=1}^m \log (t(ub(j),j)+ |x_j - ub(j)|).$$

We now describe the second lemma of \cite{derryberry2009skip}.

\begin{lemma}(Lemma 2 of \cite{derryberry2009skip})
    Executing the Skip-Splay algorithm on a query sequence $X$ of length $m$ costs $O(m\log \log n + \UB'(X)).$
\end{lemma}

The third lemma in \cite{derryberry2009skip} shows that $\UB'(X)$ is within a constant factor of Unified bound plus a linear term in $m$ greater than $\UB(X)$. The formal statement of the lemma is as follows:

\begin{lemma}(Lemma 3 of \cite{derryberry2009skip})
    For a sequence $X$ of length $m$, the following inequality is true:
    $$\sum_{j=1}^m \log (t(ub(j),j)+ |x_j- ub(j)|)\le \frac{m\pi^2 \log e}{6} + \log e + \sum_{j=1}^m 2\log (w(ub(j),j)+ |x_j- ub(j)|).$$
\end{lemma}

Note that the Lemma 3 of \cite{derryberry2009skip} is true for all $X$, so we only need to prove a lemma similar to Lemma 1 of \cite{derryberry2009skip}. We now show that the following lemma holds:

\begin{lemma}(Similar to Lemma 1 of \cite{derryberry2009skip}) \label{lem:unifiedlemma1}
    For a query sequence $X$ that is served by a group access tree $\mathcal{T}$ for Unified bound with $k$-levels, the amortized cost of the query $x_j$ using an arbitrary value of $\delta$ to define the groups is 

    $$O\left(k+ \sum_{i=1}^k \log \WS_j(g_i(x_j))\right).$$
\end{lemma}
\begin{proof}
    Let $W^i=\sum_{g'\in \dg(g)} w_j(g')$ where $g=g_{i-1}(x_j)$, then $W^i\le c$, where $c$ is a constant.

    Let $g_1=g_{i}(x_j)$. In the group access tree $\mathcal{T}$, consider the edge $e_i=(g,g_1)$. The cost of the edge $e_i$ at time $j$ is:
    $$c_{e_i}(\mathcal{T},\mathcal{W},x_j)=\log \left(\frac{c}{w_{j}(g_1)}\right)=O(\log \WS_j(g_1)).$$

    Summing over all levels proves the desired lemma.
\end{proof}

Note that as the bound in \Cref{lem:unifiedlemma1} holds for any values of $\delta$, there exists at least one value of $\delta$ such that the bound holds without randomization (same as described in \cite{derryberry2009skip}). Thus, we have proved \Cref{thm:treeub}.

\section{Unified bound with a time window}\label{sec:newubtw}

We reuse the analysis of $k$-finger to bound the time window variant of unified bound. The group access tree is the same as described for the $k$-finger bound. Let $\CC$ be a binary tree which we will use as a reference tree. We will assume that at time $t-1$, the $k$-fingers are stationed at the leaves containing $x_{t-1}, x_{t-2}, \dots, x_{t-k}$ in $\CC$. At the end of time $t$,  the finger {\em disappears} from $x_{t-k}$ and {\em appears} at $x_t$. Note that we are not moving the finger from $x_{t-k}$ to $x_t$. It just changes its position.

Similar to \Cref{sec:newkfinger}, we analyze the case when $k < n^{\frac{1}{2^j}}$. As in \Cref{sec:newkfinger}, the weight of a group may change after the finger's movement. Let $x_t$ be a key accessed at time $t$ and let $g_1=g_j(x_t)$. Then we prove the following lemma:
\begin{lemma}(Similar to \Cref{lem:newcostkfinger})
\label{lem:newcostubtw}
Assume that $k < n^{1/2^j}$. The  cost $c_{e_j}(\tset, \wset,x_t)$ at time $t$ is:

 $c_{e_j}(\tset, \wset,x_t) =
\begin{cases}
	O(\frac{1}{2^j} \log n) \hspace{0.8cm}\text{if no finger lies in  $g_1$ at time $t-1$} \\
	O(\log k) \hspace{1.2cm} \text{otherwise}
\end{cases}$
\end{lemma}
The proof is same as \Cref{lem:newcostkfinger}.

\subsection{Finishing the proof of \Cref{thm:treeubtw}}

Let $$\UB^k(t) = \min_{t-k \le t'\le t-1}\{t-t' + |x_t - x_{t'}| +2\}$$ and $$ub^k(t) = \text{argmin}_{t-k \le t'\le t-1}\{t-t' + |x_t - x_{t'}| +2\}$$ Thus, $$\UB^k(X) = \sum_{t=2}^m  \log \UB^k(t)$$ Similar to \Cref{lem:newsamegroupprobability}, we claim the following:
 
 \begin{lemma}
\label{lem:newsamegroupprobabilityubwindow}
	Let $j$ be the largest integer such that  $\UB^k(t) \le \frac{n^{1/2^j}}{\log (n^{1/2^j})}$. Then, for any $i \in [0,j]$, $$Pr\Big[ \text{$x_t$ and the $ub^k(t)$ lie in different level $i$ group}  \Big] \le \frac{cd}{\log n^{1/2^i}}$$  for some constant $c \ge 1$, $d \ge 1$.
\end{lemma}

We are now ready to bound the running time of $c(\tset, \wset,X)$. The proof is identical to the proof of  \Cref{thm:treekfinger}. Let $j$ be the largest integer such that  $\UB^k(t) \le \frac{n^{1/2^j}}{\log (n^{1/2^j})}$. At any level $\ge j+1$, the cost can be bounded as follows: 

\begin{align*}
	\sum_{i=j+1}^{\log \log n} c_{e_i}(\tset, \wset,x_t) &= \sum_{i=j+1}^{\log \log n} \max\{\frac{1}{2^i} \log n, \log k \}\\
	&= O\Big(\max\{\frac{1}{2^j} \log n, (\log \log n - j) \log k \}\Big) \\
	&= O(\max\{\log \UB^k(t), (\log \log n-j) \log k\})
\end{align*}

For the levels $i \in [0,j]$, $x_t$ and $ub^k(t)$ may lie in the same level $i$ group. The expected cost of $c_{e_i}(\tset, \wset,x_t)$ in iteration $[0,j]$ can be given using \Cref{lem:newsamegroupprobabilityubwindow} and \Cref{lem:newcostubtw}:

\begin{align*}
 E[c_{e_i}(\tset, \wset,x_t)] &= \Pr[x_t \ \text{and $ub^k(t)$ lie different level $i$ group} ] \times \frac{1}{2^i} \log n \\
& \hspace{1cm}+ \Pr[x_t \ \text{and $ub^k(t)$ lie in same level $i$ group} ] \times \log k\\
&\le \frac{c^2}{\log n^{1/2^i}} \frac{1}{2^i} \log n +  \Big(1 - \frac{c^2}{\log n^{1/2^i}}\Big) \log k\\
&\le c^2 + \log k\\
&= O(\log k)
\end{align*}

Thus, the total expected cost at time $t$ is 

\begin{align*}
\sum_{i=1}^{\log \log n} E[c_{e_i}(\tset, \wset,x_t)] &= \sum_{i=1}^{j} E[c_{e_i}(\tset, \wset,x_t)] + \sum_{i=j+1}^{\log \log n} E[c_{e_i}(\tset, \wset,x_t)] \\
\intertext{Substituting the values for these quantities, we get }
&\le (j+1) \log k + O(\max\{\log \UB^k(t) , (\log \log n - j)\log k\})  \\
&\le O(\log \log n \log k) +  O(\UB^k(t))
\end{align*}

This implies that the expected cost of $c(\tset, \wset,X)$ on the sequence $X$ is $(O(1), O(\log k \log \log n))$-competitive to the unified bound with time window of size $k$.
 Thus we have proven \Cref{thm:treeubtw}.

\section{$\GGR$ algorithm}
\label{sec:ggreedyalgo}
This section will define a new BST algorithm named $\GGR$, which satisfies the \emph{group access bound}. The $\GGR$ algorithm is very similar to $\GR$. Therefore, we first describe the $\GR$ algorithm and then provide some intuition about the $\GGR$ algorithm before describing it formally.

Given a search sequence $X=\{x_1,x_2,\dots,x_m\}$, intuitively the $\GR$ algorithm works as follows:

\noindent At time $t$, the $\GR$ algorithm performs a horizontal line sweep of the points at $y=t$.  By induction, we can assume that all the pairs of the points below the line $y=t$, are arborally satisfied.  So, at time $t$, we find all the rectangles with one endpoint $x_t$ and the other endpoint $q$, where $q$ is a point below the sweep line. If the rectangle $\square x_tq$ is not arborally satisfied, then add a point at the corner of the rectangle $\square x_tq$ on the sweep line to make it arborally satisfied. See \Cref{alg:greedy} for the formal definition of $\GR$ (\Cref{fig:gredy} illustrates the execution of $\GR$).

\begin{minipage}[t]{0.98\linewidth}
  \begin{minipage}[t]{0.44\linewidth}
  \begin{algorithm}[H]
\caption{\textsc{AddPoints}($A$,$x_t$)}
\label{alg:greedy1}
\ForEach {$\square x_tq$ in $A$}
{
	Add a point at $(q.x,x_t.y)$
}
\end{algorithm}
\end{minipage}
  \begin{minipage}[t]{0.55\linewidth}
    
    \begin{algorithm}[H]
    
\caption{Processing of $\GR$ at time $t$}
\label{alg:greedy}
Let $A$ be the set of all the arborally unsatisfied rectangles with one endpoint as $x_t$ and another endpoint below the line $y=t$\;
$\textsc{AddPoints}(A,x_t)$\;
\end{algorithm}
  \end{minipage}
\end{minipage}

\begin{figure}[hpt!]

\begin{subfigure}{0.5\textwidth}
\hspace{1.5cm}
    \begin{tikzpicture}[scale=.4]
\draw [step=1cm,gray!60!white, very thin] (1,0) grid (10,10);
\foreach \x in {1,...,10}
  \draw (\x,-1pt) node[anchor=north] {$\x$};

\draw[->][black] (1,0) -- (10.1,0);
\draw[->][black] (1,0) -- (1,10.25);
\draw[yellow,very thick] (1,0)--(1,10) (4,0)--(4,10) (7,0)--(7,10) (10,0)--(10,10);
\filldraw[red] (6,1) circle (2pt) (2,2) circle (2pt) (3,3) circle (2pt) (9,4) circle (2pt) (10,5) circle (2pt) (4,6) circle (2pt) (5,7) circle (2pt) (8,8) circle (2pt) (1,9) circle (2pt) (2,10) circle (2pt) ;


\filldraw[cyan] (4,1) circle (3pt) (7,1) circle (3pt) (1,2) circle (3pt) (4,2) circle (3pt) (1,3) circle (3pt) (4,3) circle (3pt) (7,4) circle (3pt) (10,4) circle (3pt) (7,5) circle (3pt) (1,6) circle (3pt) (7,6) circle (3pt) (4,7) circle (3pt) (7,7) circle (3pt) (7,8) circle (3pt) (10,8) circle (3pt) (4,9) circle (3pt) (1,10) circle (3pt) (4,10) circle (3pt);

\filldraw[black] (4,4) circle (3pt) (7,9) circle (3pt);
\end{tikzpicture}
\captionsetup{justification=centering}
    \caption{$\GGR$}
    \label{fig:mgreedy}
\end{subfigure}
\begin{subfigure}{0.4\textwidth}
\hspace{.8cm}
   \begin{tikzpicture}[scale=.4]
\draw [step=1cm,gray!60!white, very thin] (1,0) grid (10,10);
\foreach \x in {1,...,10}
  \draw (\x,-1pt) node[anchor=north] {$\x$};

\draw[->][black] (1,0) -- (10.1,0);
\draw[->][black] (1,0) -- (1,10.25);

\filldraw[red] (6,1) circle (2pt) (2,2) circle (2pt) (3,3) circle (2pt) (9,4) circle (2pt) (10,5) circle (2pt) (4,6) circle (2pt) (5,7) circle (2pt) (8,8) circle (2pt) (1,9) circle (2pt) (2,10) circle (2pt) ;

\filldraw[blue] (6,2) circle (2pt) (2,3) circle (2pt) (6,3) circle (2pt) (6,4) circle (2pt) (9,5) circle (2pt) (3,6) circle (2pt) (6,6) circle (2pt) (9,6) circle (2pt) (4,7) circle (2pt) (6,7) circle (2pt) (6,8) circle (2pt) (9,8) circle (2pt) (2,9) circle (2pt) (3,9) circle (2pt) (4,9) circle (2pt) (6,9) circle (2pt) ;
\end{tikzpicture}
\captionsetup{justification=centering}
\caption{$\GR$}
    \label{fig:gredy}
    
\end{subfigure}
\caption{$\GGR$(on level 1) vs $\GR$ on a sequence $(6,2,3,9,10,4,5,8,1,2)$. There are three groups $\{[1 \dots 4], [4 \dots 7], [7 \dots 10] \}$ in level 1, recursively we can define groups within each group. Red points are the searched keys. (a) Light blue points are added at the boundary of a group $g=g_1(x_t)$. Black points are added to make $\GGR$ arborally satisfied with other groups in level 1 (b) Blue points are added to make $\GR$ arborally satisfied.}
\end{figure}
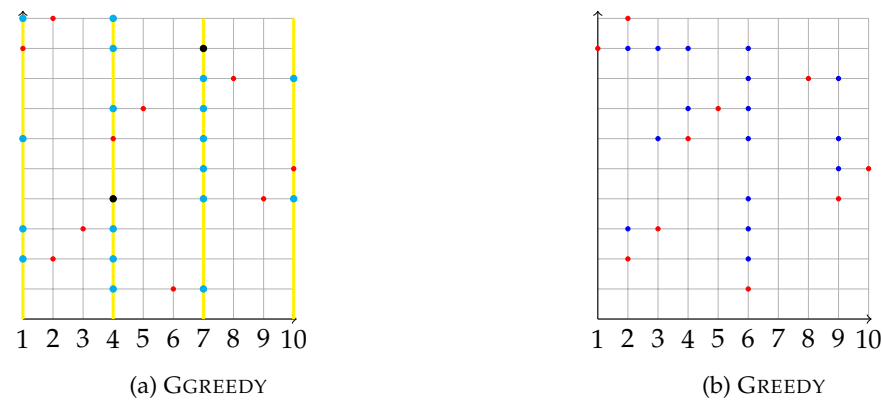

We will now try to provide some intuition (motivation) behind the $\GGR$ algorithm. While analyzing the $\GR$ algorithm, we observed that if we can partition the keys into groups (containing consecutive keys) and treat these groups as keys, then we can apply the $\GR$ algorithm to this new set of keys. We can then recursively do this within each group, making the analysis easy.

As an example,  we can divide the key space into three groups, say from $[1,\dots,n/3],[n/3+1,\dots,2n/3]$ and $[2n/3+1,\dots,n]$ which we can represent as keys $k_1$, $k_2$ and $k_3$. We can now run the $\GR$ algorithm on this set of keys. When a key in the range $[n/3+1,\dots,2n/3]$ is searched at time $t$, we can assume that the key $k_2$ is searched and perform the $\GR$ algorithm accordingly on the keys $k_1$, $k_2$ and $k_3$. We can recursively do this procedure inside each group. To separate the groups and to ensure that the resultant point set is \emph{arborally satisfied}, we add points at the boundaries of the groups.

One can observe that the recursive partition of the keys corresponds to a \emph{group access tree}.  An astute reader can see that $\GGR$ is not a single algorithm but a family of BST algorithms, which depends on the group access tree.

We are now ready to define the $\GGR$ algorithm. Pick a group access tree $\tset$. Let $g$ be a group at any level in the group access tree. The group $g$ contains a set of consecutive keys at the leaves of $\tset$. The group $g$ has two end keys, which we denote as the boundary keys of a group.


\begin{definition}(Boundary of a group)

\noindent Let $g$ be a group that contains consecutive keys $(a,a+1,a+2,\dots,b)
$, then $\LE(g) = a$ is called the left boundary  key 
of the group $g$. Similarly, $\RI(g) = b$ is called the right boundary key. 
Together, they are called the boundary keys of the group $g$.
\end{definition}

In $\GR$, we access a search key $x_t$ at time $t$, but in $\GGR$, when a key is searched at time $t$, we access a group per level of the group access tree $\TT$. We define the access of a group as follows:

\begin{definition}(Accessed group at time $t$ in level $j$)
  A group $g$ is said to be accessed at time $t$ in level $j$ if $g=g_{j}(x_t)$.  
\end{definition}

Note that the searched key $x_t$ lies in the interval of $g$ at time $t$ in level $j$.

\noindent While accessing a key in $\GR$ at time $t$, the algorithm touches keys to form an arborally satisfied set. Similarly, in $\GGR$, when a group is accessed at time $t$, we might touch other groups. We define a touch group in $\GGR$ as follows:

\begin{definition}(Touched group at time $t$ in level $j$)
    A group $g$ is said to be touched at time $t$ in level $j$ if a boundary key(s) of $g$ is touched.
\end{definition}
\begin{remark}
    An accessed group is also a touch group at time $t$ in level $j$ where both the boundary keys are touched.
\end{remark}
Similar to the definition of \emph{unsatisfied rectangle} for $\GR$, we define an \emph{unsatisfied group} in $\GGR$ as follows:
\begin{definition}(Unsatisfied group at time $t$ in level $j$)
    A group $g_1$ is said to be unsatisfied at time $t$ in level $j$ if the boundary key(s) of $g_1$ form an unsatisfied rectangle with the boundary keys of $g=g_j(x_t)$ at time $t$. We denote the unsatisfied group by $\square g_1g$.
\end{definition}

Let $g_1$ be a group at time $t$ in level $j$ which is unsatisfied when the group $g$ is accessed. We touch $g_1$ at time $t$ in level $j$ to make it an \emph{arborally satisfied} group.

Let us now informally describe the $\GGR$ algorithm. When a key $x_t$ is accessed at time $t$, we access one group per level in the group access tree $\tset$. While accessing the group $g=g_j(x_t)$ in level $j$, we find all the unsatisfied groups in level $j$ which lie inside the group $g_{j-1}(x_t)$ and make them arborally satisfied. We recursively do this for level $j+1$ and so on (See \Cref{fig:ggreedyexe}). The $\GGR$ algorithm can be viewed as if we are applying the $\GR$ algorithm on groups at each level of the group access tree $\tset$. Hence the name, $\GR$ on groups or $\GGR$.

\begin{figure}[htp!]
  \centering

  \includegraphics[trim={0 120 10 70},clip,scale=.4]{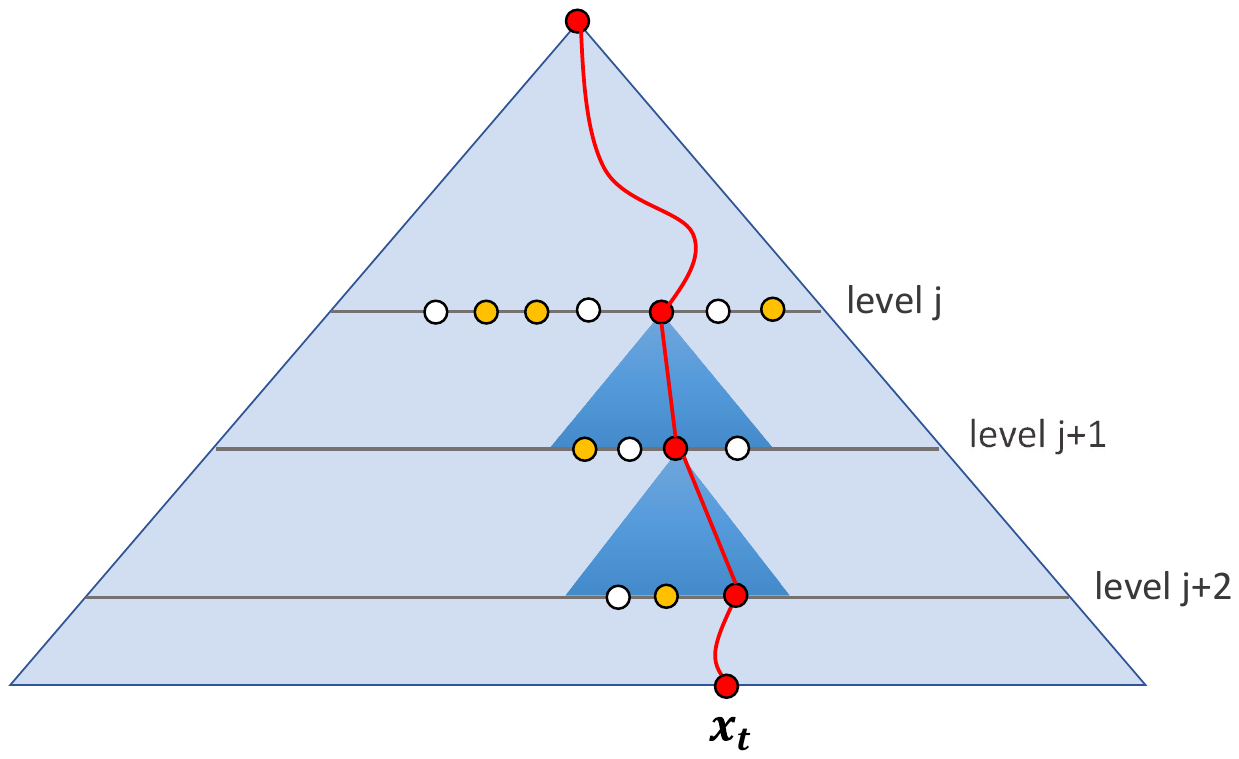}
  \caption{An example when $\GGR$ accesses $x_t$. Red nodes are the accessed groups. Orange nodes are the touched groups. White nodes denote the non-touched groups.} 
  \label{fig:ggreedyexe}
\end{figure}

We now describe the $\GGR$ algorithm formally (See \Cref{alg:GGR}).

\begin{minipage}[t]{0.98\linewidth}
  \begin{minipage}[t]{0.44\linewidth}
  \begin{algorithm}[H]
\caption{Touchgroups($A$,$g$,$j$)}
Touch $g$ in level $j$\;
\ForEach {$\square g_1g$ in $A$}
{
	Touch $g_1$ in level $j$\;
}
\end{algorithm}
\end{minipage}
  \begin{minipage}[t]{0.55\linewidth}
    
    \begin{algorithm}[H]
\caption{Processing of $\GGR$ at time $t$}
\label{alg:GGR}
\ForEach{$j=1$ to $k$}
{
Let $A_j$ be the set of all arborally unsatisfied groups in level $j$ when accessing $g=g_j(x_t)$\;
Touchgroups$(A_j,g,j)$\;
}
\end{algorithm}
  \end{minipage}
\end{minipage}

In the above algorithm, at iteration $j$, we first add points at $\LE(g_{j}(x_t))$ and $\RI(g_{j}(x_t))$ (same as touching group $g_{j}(x_t)$). Then, we process all the arborally unsatisfied groups with one endpoint as $g_j(x_t)$ and the other endpoint inside the group $g_{j-1}(x_t)$ (See \Cref{fig:mgreedy} for the execution of $\GGR$ on level 1). 

We will now try to bound the number of groups touched by $\GGR$ in iteration $j$ of Algorithm \ref{alg:GGR}. Let us give it a special notation:

\begin{definition}
	Let $\TT_j(t)$ be the set of groups touched by $\GGR$ in the $j$-th iteration of Algorithm \ref{alg:GGR} at time $t$.  
\end{definition}

Touching a group is the same as touching the group's boundary key(s). Therefore, we consider the boundary keys of the group for the rest of this section to provide better-detailed proofs. We now show some essential properties of $\TT_j(t)$.

\begin{lemma}
\label{lem:iterj}
	$\TT_j(t)$ only contains points that are at the boundaries of level $j$ groups that lie inside $g_{j-1}(x_t)$.
\end{lemma}
\begin{proof}
\begin{figure}[hpt!]
\centering

   \begin{tikzpicture}[scale=0.5]
\draw[->][black] (0,-1) -- (13.1,-1);
\draw[->][black] (0,-1) -- (0,8.1);

\filldraw[green!30!white] (2,7) rectangle (10,5);
\filldraw[yellow!50!white] (4,4) rectangle (10,2);
\draw[olive!60!white,very thick] (1,0)--(1,8) (5,0)--(5,8);
\draw[orange,very thick] (9,0)--(9,8) (13,0)--(13,8);
\filldraw[red] (2,7) node[anchor=south]{$x_t$} circle (2pt) (4,4)node[anchor=south]{$x_{t_1}$} circle (2pt) (12,2) node[anchor=south]{$x_{t_2}$} circle (2pt) ;

\filldraw[blue] (10,2) node[anchor=west]{$r$} circle (2pt) (10,5) node[anchor=west]{$q$} circle (2pt);
\filldraw[cyan] (9,2) circle (3pt) (13,2) circle (3pt);
\filldraw[gray!75!white](10,4) node[anchor=west]{$p_1$} circle (2pt) (10,7) node[anchor=west]{$p
$} circle (2pt);
\draw [<->,,black,thick] (1,0)--(5,0); 
\draw [<->,,black,thick] (9,0)--(13,0); \draw[black](3,0)node[anchor=north]{$\GG_1$} (11,0)node[anchor=north]{$\GG_2$};
\end{tikzpicture}

\caption{$\GGR$ adds only boundary points.}
    \label{fig:endpoints}

\end{figure}
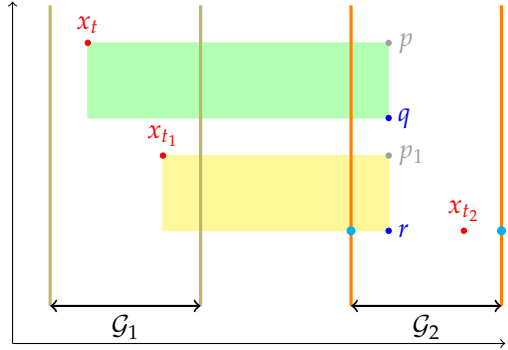

From the construction of $\TT_j(t)$, it is trivial to see that all its points lie inside $g_{j-1}(x_t)$. Let us now show that they will lie only at the boundaries of level $j$ groups.

Let $\GG_1$ and $\GG_2$ be two groups at level $j$ that lie inside $g_{j-1}(x_{t})$. Let us assume $x_t \in \GG_1$. We will now show that $\TT_j(t)$ contains points only at the boundary of $\GG_2$ but not inside it. For contradiction, assume that $\TT_j(t)$ contains a point $p$ that lies inside $\GG_2$ (See \Cref{fig:endpoints}). This implies that $\GGR$ encountered an unsatisfied rectangle, say $\square x_tq$. Let $x_{t_1}$ be the smallest $t_1 <t$ such that $\GGR$ added a point, say $p_{1}$, at $(p_1.x,t_1)$ and $x_{t_1}$ lies outside the group $\GG_2$. This implies that $\GGR$ encountered an unsatisfied rectangle $\square x_{t_1}r$ at time $t_1$. Note that the point $r$ must exist.  If $r.y = t_2$, then by construction $x_{t_2}$ lies in the group $\GG_2$. But at time $t_2$, $\GGR$ would have added points at the boundaries of $\GG_2$. This implies that $\GGR$ will not add the point $p_1$ at time $t_1$ as the rectangle $\square x_{t_1}r$ is already satisfied by one of the endpoints of the group. This leads to a contradiction proving the lemma.
\end{proof}

 An immediate corollary of the above lemma is:

\begin{corollary}
\label{cor:nottouch}
	Let $p$ be a point that lies strictly inside $g_{j-1}(x_t)$. If $\GGR$ adds $p$ at time $t$ then $x_t$  lies in $g_j(p)$.
\end{corollary}

We will now show that $\GGR$ outputs an arborally satisfied set. Let us assume that we are processing $x_t$. We can visualize $\GGR$ as running $\GR$ first on the level 1 group. This execution is the same as $\GR$. Then, we add points at the boundary of  $g_1(x_t)$. This is the first step when we deviate from the $\GR$ algorithm. Once we have added the boundary points, we again run $\GR$ on level 2 groups, and the process continues. One may feel that touching the boundary keys may create some unsatisfied rectangles. But in the following lemma, we show that this is not the case.
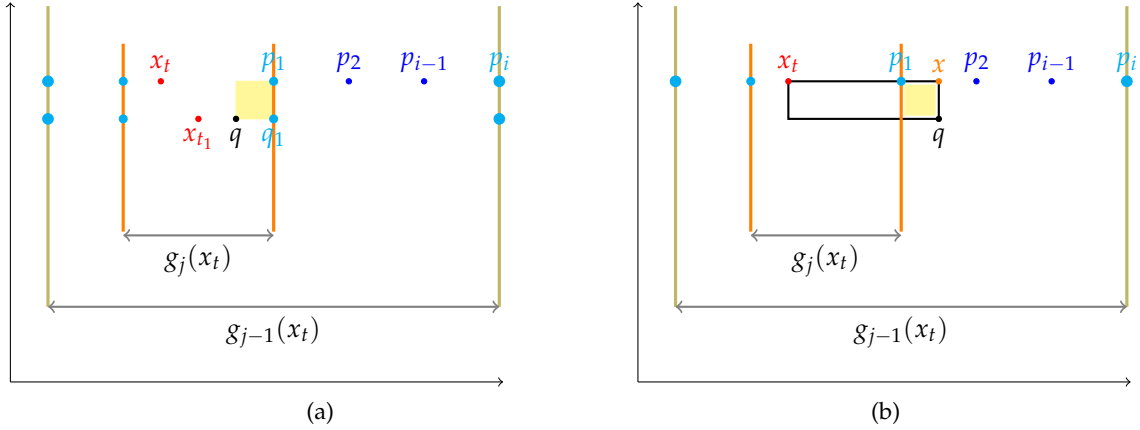
\begin{figure}[hpt!]
\centering
\begin{subfigure}{0.5\textwidth}
   \begin{tikzpicture}[scale=.5]
\draw[->][black] (0,-2) -- (13.1,-2);
\draw[->][black] (0,-2) -- (0,8.1);

\filldraw[yellow!50!white] (7,6) rectangle (6,5);
\draw[olive!60!white,very thick] (1,0)--(1,8) (13,0)--(13,8);
\draw[orange,very thick] (3,2)--(3,7) (7,2)--(7,7);
\filldraw[red] (4,6) node[anchor=south]{$x_t$} circle (2pt) (5,5) node[anchor=north]{$x_{t_1}$} circle (2pt) ;
\filldraw[black] (6,5) node[anchor=north]{$q$} circle (2pt);

\filldraw[blue] (9,6) node[anchor=south]{$p_2$} circle (2pt) (11,6) node[anchor=south]{$p_{i-1}$} circle (2pt);
\filldraw[cyan] (7,6) circle (3pt)  node[anchor=south]{$p_1$} (3,6) circle (3pt) (1,6) circle (4pt) (13,6)  node[anchor=south]{$p_i$} circle (4pt) (7,5) circle (3pt)  node[anchor=north]{$q_1$} (3,5) circle (3pt) (1,5) circle (4pt) (13,5) circle (4pt);

\draw [<->,gray,thick] (3,1.9)--(7,1.9); \draw[black](5,1.9)node[anchor=north]{$g_j(x_t)$};
\draw [<->,gray,thick] (1,0)--(13,0); \draw[black](7,0)node[anchor=north]{$g_{j-1}(x_t)$};
\end{tikzpicture}
\caption{}
    \label{fig:ggreedyarbo2}
    
\end{subfigure}
\begin{subfigure}{0.4\textwidth}
     \begin{tikzpicture}[scale=.5]
\draw[->][black] (0,-2) -- (13.1,-2);
\draw[->][black] (0,-2) -- (0,8.1);

\draw[black,thick] (4,6) rectangle (8,5);
\filldraw[yellow!50!white] (7.1,5.9) rectangle (7.9,5.1);
\draw[olive!60!white,very thick] (1,0)--(1,8) (13,0)--(13,8);
\draw[orange,very thick] (3,2)--(3,7) (7,2)--(7,7);
\filldraw[red] (4,6) node[anchor=south]{$x_t$} circle (2pt) ;
\filldraw[black] (8,5) node[anchor=north]{$q$} circle (2pt);
\filldraw[orange] (8,6) node[anchor=south]{$x$} circle (2pt);

\filldraw[blue] (9,6) node[anchor=south]{$p_2$} circle (2pt) (11,6) node[anchor=south]{$p_{i-1}$} circle (2pt);
\filldraw[cyan] (7,6) circle (3pt)  node[anchor=south]{$p_1$} (3,6) circle (3pt) (1,6) circle (4pt) (13,6)  node[anchor=south]{$p_i$} circle (4pt);

\draw [<->,gray,thick] (3,1.9)--(7,1.9); \draw[black](5,1.9)node[anchor=north]{$g_j(x_t)$};
\draw [<->,gray,thick] (1,0)--(13,0); \draw[black](7,0)node[anchor=north]{$g_{j-1}(x_t)$};
\end{tikzpicture}
    \caption{}
    \label{fig:ggreedyarbo1}
\end{subfigure}
\caption{Illustration of the proof of \Cref{lem:GGarborally}. Figure (a) depicts the case when the point $q$ is inside $g_j(x_t)$. Figure (b) depicts the case when the point $q$ is outside $g_j(x_t)$. }
\label{fig:ggreedyarbo}
\end{figure}
\begin{lemma}\label{lem:GGarborally}
$\GGR$ outputs an arborally satisfied set on any input representing $X$.
\end{lemma}
\begin{proof}
We will show that $\GGR$ outputs an arborally satisfied set for each prefix of the point set representing $X$. For the base case, $\GGR$ is arborally satisfied at time $t=1$ trivially. Using the induction hypothesis, let us assume that $\GGR$ outputs an arborally satisfied set at time $t-1$. We will show that $\GGR$ outputs an arborally satisfied set at time $t$.

To this end, we will use induction a second time to show that $\GGR$ does not create any arborally unsatisfied rectangle due to newly added points in $\TT_j(t)$ where $j \in [1 \dots k]$. 
 We skip the base case as it is similar to the argument we present below. Let us assume for contradiction that there exists some unsatisfied rectangle after iteration $j$. 

Let $p_1,p_2,\dots,p_{i-1}$ be the points in $\TT_j(t)$ to the right of $x_t$ where $\RI(g_j(x_t)) = p_1$. Let us also assume that $p_{i} = \RI(g_{j-1}(x_t))$. If $j=1$, we can assume $p_{i}$ to be a dummy key $n+1$.   
  Since $\GGR$ outputs an arborally satisfied set at time $t-1$, one of the endpoints of the arborally unsatisfied rectangle is a newly added point, say $p_k$. There are two cases:
\begin{enumerate}
	\item $k = 1$
	
	$p_1$ is a boundary point. Let us assume that there is an unsatisfied rectangle $\square p_1q$ such that $q$ lies in $g_j(x_t)$ but is not the left boundary point of $g_j(x_t)$. By \Cref{cor:nottouch}, point $q$ can be touched only while processing a search key that lies in $g_j(x_t)$. Let us assume that the point $q$ is added at time $t_1$ where $t_1 <t$. But then, we would have added another point, say $q_1$, at the right boundary of $g_j(x_t)$ even at time $t_1$. Thus, $\square p_1q$ is already satisfied due to $q_1$ (See \Cref{fig:ggreedyarbo}(a)). If point $q$ is the left boundary key of $g_j(x_t)$ or lies outside $g_j(x_t)$ to the left of $p_1$, then $\square p_1q$ is satisfied by $x_t$ itself.

	Let $q$ be the nearest point to the right of $p_1$ such that $\square p_1q$ is an unsatisfied rectangle. Thus, $\square p_1q$ does not contain any point.  We will now show that even $\square x_tq$ does not contain any point and therefore $\GGR$ will add a point at the corner of $\square x_tq$ at time $t$. This will lead to a contradiction. 
	 
	 Assume that $q$ was added at time $t_1$ due to search $x_{t_1}$. $x_{t_1}$ cannot lie in $g_j(x_t)$ as otherwise $\square p_1q$ will be satisfied by $right(g_j(x_{t_1}))$. In fact, from time $t_1$ to $t-1$ no search can lie inside $g_j(x_t)$ because the right boundary key of $g_j(x_{t'})$ ($t_1\le t' \le t-1$) will make $\square p_1q$ arborally satisfied. As there are no searches inside $g_j(x_t)$ from $t_1$ to $t-1$ no points are added inside $g_j(x_t)$ during this time interval. Using the fact that there are no points in $\square p_1q$, we claim that  $\square x_tq$ is also unsatisfied at time $t$ and $\GGR$ will add a point in the corner of $\square x_tq$.

	\item $k \neq 1$
	
	Let us assume there is an unsatisfied rectangle $\square p_kq$. If $q$ lies in $g_j(x_t)$, then $\square p_kq$ is already satisfied by $p_1$. So,  $q$ lies outside $g_j(x_t)$. If $q$ lies to the left of $p_k$, then the unsatisfied rectangle for which $p_k$ was added at time $t$ is already satisfied by $q$. Therefore $q$ lies to the right of $p_k$. As in the above point, we can now show that in this case also, $\square x_tq$ will be an unsatisfied rectangle.

\end{enumerate}

This completes the proof of the lemma.
\end{proof}



In the next section, we generalize the \emph{access lemma}, which have been proved for Splay tree and $\GR$ \cite{tarjan1985sequential,fox11}.  Let us define a notation before we move to the next section. 

\begin{definition}
	Let $\HTT_j(t)$ denote the amortized number of groups touched by $\GGR$ in the $j$-th iteration of \Cref{alg:GGR} at time $t$.
\end{definition}

\section{The Group Access Lemma}\label{sec:newgroupaccess}
In this section, we introduce the \emph{group access lemma}, which is a generalization of the 
{\em access lemma} and show that the $\GGR$ algorithm satisfies it.

Consider the group access tree $\tset$. When a BST algorithm $\aset$ accesses $x_t$ at time $t$, after the access, it adjusts itself to $\aset'$. Let $g_1=g_j(x_t)$ be the group that contains  $x_t$. Let $\Phi_t^j$ be a potential function that depends on the state of the algorithm with respect to the groups at level $j$ in $\tset$. Define $\Phi_t = \sum_j \Phi^j_t$. Define the group access lemma as follows:

\begin{definition}(Group Access Lemma)
 A BST algorithm $\aset$ satisfies the group access lemma if the amortized cost to access $x_t$  at time $t$ in level $j$ is: $$\HTT_j(t)\le O\left(\log \frac{W^j}{w_{t-1}(g_1)} \right)+ \Phi_{t-1}^j - \Phi_{t}^j$$

\end{definition}
The amortized cost of the algorithm $\aset$ at time $t$ can be denoted as:

    $$\HTT(\aset,x_t)=\sum_j \HTT_j(t)$$
and the amortized cost of the algorithm $\aset$ on the access sequence $X$ can be defined as:
  $$\HTT(\aset,X)=\sum_t \HTT(\aset,x_t).$$
 
With the definition of the group access lemma in hand, we will now show that $\GGR$ satisfies the group access lemma.

\subsection{\GGR satisfies the group access lemma }

The proof of the group access lemma for $\GGR$ is similar to the proof of \emph{access lemma} derived in \cite{chalermsook2015self}. 

Let $x_t$ be the searched key at time $t$ and let $g"=g_{j-1}(x_t)$. Remember that  $\dg(g")$ denotes the child groups of $g"$. Let $g_1=g_j(x_t)$. Following \cite{fox11}, we define the {\emph neighborhood} of a group in $\dg(g")$.  For $g \in \dg(g")$, let $\rho(g,t)$ be the last access of group $g$ at or before time $t$. Let $g'  \in \dg(g")$ be the first group to the left of $g$ such that $\rho(g',t) \ge \rho(g,t)$. If such a $g'$ does not exist, then we set $g'$ to be the leftmost group in $\dg(g")$. The \emph{left neighborhood} of $g$ at time $t$ is all the groups strictly between $g$ and $g'$. The left neighborhood is denoted by $\nset_l(g,t)$. Similarly, define the \emph{right neighborhood}  of $g$ and denote it as $\nset_r(g,t)$. The neighborhood of $g$ at time $t$ is $\nset(g,t)=\nset_l(g,t) \cup \{g\}\cup \nset_r(g,t)$. 
 
 Let $w^j(\nset(g,t))$ be the sum of the weights of the groups in the neighborhood of $g$ at time $t$ in level $j$. Let $\Phi_t^j(g)=\log(w^j(\nset(g,t)))$ denote the potential of $g$. We abuse this notation and extend it to a set of groups. Thus, $\Phi_t^j(\dg(g"))=\sum_{g \in \dg(g")}\log(w^j(\nset(g,t)))$.

Let $\TT_j(t)$ be the groups touched in level $j$ among $\dg(g")$ at time $t$. The amortized cost $\HTT_j(t)$ at time $t$ in level $j$ is given by
$$\HTT_j(t)=\TT_j(t)+\Phi_{t-1}^j(\dg(g"))-\Phi_t^j(\dg(g"))$$
Note that here the change in potential indicates the change in the neighborhood of the groups in $\dg(g")$ from time $t-1$ to time $t$. In the rest of the section, we will bound $\TT_j(t)$. To this end, we first bound the number of groups touched to the left of $g_1$ at time $t$. In a similar way, we can bound the number of groups touched to the right of $g_1$ at time $t$.

Let $Y=\{g^1,g^2,\dots,g^q\}$ be the groups to the left of $g_1$ that are touched in $\dg(g")$ where $g^1$ lies to the left of $g_1$, $g^2$ lies to the left of $g^1$ and so on. Thus, before time $t$, $g_1$ lies in the neighborhood of $g^1$, $g^1$ lies in the neighborhood of $g^2$, and so on. 

Let $w_0=(w^j(\nset(g_1,t-1)))$.
Let $a_1$ be the leftmost group in $Y$ such that $w_{t-1}(a_1)\le 2w_0$. In general, $a_i$ is the leftmost group in $Y$ such that $w_{t-1}(a_i)\le 2^iw_0$. It is easy to observe that the set $\{a_1, a_2, \dots\}$ can have atmost $\log\left(\frac{W^j}{w_0}\right)$ distinct groups.

Let $b$ be a group between $a_i$ and $a_{i+1}$ such that $w^j(\nset(b,t))> 2^{i-1}w_0$. We call $b$ a \emph{heavy group}.
There can be atmost 3 heavy groups between $a_i$ and $a_{i+1}$ as otherwise $w_{t-1}(a_{i+1})>2^{i+1}w_0$, which is a contradiction.

Now, we count the number of light (not heavy) groups between $a_i$ and $a_{i+1}$. Let $b'$ be a light group, then $w^j(\nset(b',t))\le 2^{i-1}w_0$. Also, we know that $w^j(\nset(a_{i+1},t-1))
> 2^iw_0$. Thus the ratio $r_i=w^j(\nset(g^{i+1},t-1))/w^j(\nset(g^i,t))\ge 2$, if $g^i$ is a light group, otherwise $r_i\ge 1$ for $1\le i \le q-1$. Therefore,
$$2^{\text{no. of light groups}}\le \prod_{1\le i \le q-1}r_i=\left(\prod_{1\le i \le q} \frac{w^j(\nset(g^i,t-1))}{w^j(\nset(g^i,t))}\right).\frac{w^j(\nset(g^q,t))}{w_0}$$

Taking $\log$ on both sides, we get that:

\begin{tabbing}
$\text{no. of light groups}$\= $\le$ \= $\sum_{i=1}^{q} \log \left(w^j(\nset(g^i,t-1))\right) - \sum_{i=1}^{q}\log\left(w^j(\nset(g^i,t))\right) + \log\left(\frac{w^j(\nset(g^q,t))}{w_0}\right)$ \\
 \> $=$ \> $\Phi_{t-1}^j(Y)-\Phi_t^j(Y) + \log\left(\frac{W^j}{w_0}\right)$
\end{tabbing}

Therefore,
$$|Y|\le 4\log\left(\frac{W^j}{w_0}\right)+ \Phi_{t-1}^j(Y)-\Phi_t^j(Y)$$
Similarly, we can bound the number of groups touched to the right of $g_1$. Hence, we have shown that $\GGR$ satisfies the group access lemma.

\section{Simulation Theorem}\label{sec:newggreedycomp}
In this section, we show that the cost of $\GGR$ is bounded by the cost of the group access bound on any group access tree $\tset$. This will imply that $\GGR$ satisfies \Cref{thm:simulation}.

 As $\GGR$ satisfies the group access lemma, the cost of $\GGR$ on an access sequence $X$ without the change in potential satisfies \Cref{thm:simulation}. Let $x_t$ be searched at time $t$ and let $g =g_{j-1}(x_t)$ and $g_1 =g_j(x_t)$. In this section, we will try to bound the change in the potential for $\GGR$. To this end, we divide the execution of $\GGR$ into two parts at time $t$ by introducing an intermediate step $t'$.

\begin{enumerate}
	\item $t-1$ to $t'$
	
	In the first part, we assume that $\GGR$ has added the points in $\TT_j(t)$, but the weights have stayed the same. This is an intermediate step, which we denote as $t'$. 
	
	\item $t'$ to $t$
	
	In the second part, we will change the weights of the groups. This may cause a change in the potential for some groups. 
\end{enumerate}

\subsection{From $t-1$ to $t'$}\label{sec:neighborhood}

As $\GGR$ satisfies the group access lemma, the cost at time $t'$ in level $j$ is:
\begin{equation}\label{eq:cost}
\Phi_{t-1}^j(\dg(g)) - \Phi_{t'}^j(\dg(g)) + \log\left(\frac{W^j}{w_{t-1}(g_1)}\right)\ge \HTT_j(t')
\end{equation}

\subsection{From $t'$ to $t$}

We will now try to bound the amortized cost due to the change in weights from time $t'$ to $t$.  We need to bound the change in the potential of the groups in $\dg(g)$ i.e., 
 $\Phi_{t'}^j(\dg(g)) - \Phi_{t}^j(\dg(g))$.

We want to bound the decrease in potential from time $t'$ to $t$ in level $j$. As the weight function is locally bounded, we change the weight of each group in such a way that only the weight of $g_1$ increases from $t'$ to $t$, i.e., $w_{t}(g_1) \ge w_{t'}(g_1)$. In fact, this is true for all the bounds we have proved for the group access tree $\tset$ (See \Cref{tab:inc}).   Thus, the potential of $g_1$ may decrease. In the following lemma, we bound the potential change for $g_1$.

\begin{lemma}\label{lem:potentialdec}
$\Phi_{t'}^j(g_1) - \Phi_{t}^j(g_1) + \log \left(\frac{W^j}{w_{t-1}(g_1)}\right) \ge 0 $.
\end{lemma}
\begin{proof}
    At time $t'$, 
    $$\Phi_{t'}^j(g_1)=\log(w^j(\nset(g_1,t')))\ge \log(w_{t'}(g_1))$$
    and at time $t$,
    $$\Phi_{t}^j(g_1)=\log(w^j(\nset(g_1,t)))\le \log(W^j)$$
    Therefore,
    $$\Phi_{t'}^j(g_1) - \Phi_{t}^j(g_1)\ge \log(w_{t'}(g_1))-\log(W^j)=-\log \left(\frac{W^j}{w_{t'}(g_1)}\right)=-\log \left(\frac{W^j}{w_{t-1}(g_1)}\right)$$
\end{proof}

\begin{table}[h!]
\centering

\begin{tabular}{||c|c|c||} 
 \hline
 Bound & $t-1$ & $t$\\ [0.5ex] 
 \hline\hline
 $\sqrt{\log n}$ & 1 & $2^{\sqrt{\log n}}$ \\ 
 \hline
 k-finger & $\frac{1}{n^{1/2^j}}$ & $\frac{1}{k}$\\
 \hline
 Unified bound & $\frac{1}{(\WS_{t-1}(g_1))^2}$ & 1\\
 \hline
 Unified bound time window  & $\frac{1}{n^{1/2^j}}$ & $\frac{1}{k}$\\ [1ex] 
 \hline

\end{tabular}
\caption{Weight of $g_1$ increases from time $t-1$ to $t$.}
\label{tab:inc}
 \end{table}

 For all the other groups $g' \in \dg(g) \setminus g_1$, we do not increase the weight of $g'$  from time $t'$ to $t$. This is true for all the bounds we have shown for the group access tree $\tset$ (See \Cref{tab:dec}). Thus, $w_{t}(g') \le w_{t'}(g')$. So, the weight of only $g_1$ can increase in $\dg(g)$. However, $g_1$ does not lie in the neighborhood of any other group at time $t$ in level $j$ as it lies in $\TT_j(t)$. Therefore, for any $g' \in \dg(g) \setminus g_1$, the potential of group $g'$ can only increase. Hence, we can claim the following corollary:
\begin{corollary}\label{cor:potentai}
 In $\dg(g)$, as the potential of only $g_1$ decreases we get that $\Phi_{t'}^j(\dg(g)) - \Phi_{t}^j(\dg(g)) + \log \left(\frac{W^j}{w_{t-1}(g_1)}\right) \ge 0$.
\end{corollary}

Using \Cref{eq:cost} and \Cref{cor:potentai}, we get:
$$\Phi_{t-1}^j(\dg(g)) - \Phi_{t}^j(\dg(g)) + 2\log\left(\frac{W^j}{w_{t-1}(g_1)}\right)\ge \HTT_j(t)$$

But $\log (\frac{W^j}{w_{t-1}(g_1)})$ can be written as $c_{e_j}(\tset, \wset,x_t))$. So we get,

$$\Phi_{t-1}^j(\dg(g)) - \Phi_{t}^j(\dg(g)) + 2c_{e_j}(\tset, \wset,x_t)\ge \HTT_j(t)$$

Summing over all levels $j$ we get:

$$\sum_j\Phi_{t-1}^j(\dg(g)) - \sum_j\Phi_{t}^j(\dg(g) + 2c(\tset, \wset,x_t))\ge \HTT(\GGR,x_t)$$

Summing over all $t$ we get:

$$\sum_t\sum_j\Phi_{t-1}^j(\dg(g)) - \sum_t\sum_j\Phi_{t}^j(\dg(g)) + 2c(\tset, \wset,X)\ge \HTT(\GGR,X)$$

As the change in potential forms a telescoping series we get:

$$  \HTT(\GGR,X) \le \Phi_0 - \Phi_t + O(c(\tset, \wset,X))$$

\begin{table}
\centering

\begin{tabular}{||c|c|c||} 
 \hline
 Bound & $t-1$ & $t$\\ [0.5ex] 
 \hline\hline
 $\sqrt{\log n}$ & $2^{\sqrt{\log n}}$ & 1 \\ 
 \hline
 k-finger & $\frac{1}{k}$ & $\frac{1}{n^{1/2^j}}$ \\
 \hline
 Unified bound & $\frac{1}{(t_1)^2}$ & $\frac{1}{(t_1+1)^2}$\\
 \hline
 Unified bound time window  & $\frac{1}{k}$ & $\frac{1}{n^{1/2^j}}$\\ [1ex] 
 \hline

\end{tabular}
\caption{Weight of $g' \in \dg(g) \setminus g_1$ does not increase from time $t-1$ to $t$. Here $t_1=\WS_{t-1}(g')$.}
 \label{tab:dec}
 \end{table}

Next, we calculate the initial and the final potential. Let us assume that the weight of each group at any level in the group access tree is bounded by $[1/M, M]$ where $M$ is a polynomial in $n$, i.e., $M = n^c$, for some constant $c\ge 1$.
The reader can check that this fact is actual for all the weights that we have assigned in the group access tree. Consider level $j$ in the group access tree $\tset$. There can be at most $n$ groups at level $j$ (as there is only $n$ singleton). Thus, the potential of each group at level $j$ is bounded by $\log(nM) = O(\log n)$. Therefore,  the potential at time 0 in level $j$ can be bounded as $\Phi^j_0 = O(n \log n)$. Since there can be at most $\log n$ levels in the group access tree, $\Phi_0 = O(n \log^2 n)$. Similarly, we can bound $-\Phi_t = O(n \log^2 n)$. 

 Let $m=|X|$ such that $m=\Omega(n \log^2 n)$. Then, the amortized cost of $\GGR$ on an access  sequence $X$ of length $m$ is:

$$
 \HTT(\GGR,X) = O(c(\tset, \wset,X) +m)
$$
Therefore, $\GGR$ satisfies \Cref{thm:simulation}.

\bibliographystyle{alpha}
\bibliography{references}

\end{document}